%% file: iclr2025_conference.tex
\documentclass{article} 
\usepackage{iclr2025_conference,times}

\input{math_commands.tex}

\usepackage{hyperref}
\usepackage{url}
\usepackage{graphicx} 
\usepackage{amsfonts}
\usepackage{array}
\usepackage{graphicx}
\usepackage{amsmath}
\usepackage{comment}
\usepackage{algorithm}
\usepackage{algorithmic}
\usepackage{cancel}
\usepackage{booktabs}
\usepackage{enumerate}
\usepackage{amsthm}
 \usepackage{amssymb}
\usepackage{makecell}
\usepackage[caption=false,font=normalsize,labelfont=sf,textfont=sf]{subfig}
\usepackage[export]{adjustbox}
\usepackage[T1]{fontenc}

\newtheorem{theorem}{Theorem}
\usepackage{multirow} 
\usepackage[utf8]{inputenc} 
\usepackage[T1]{fontenc}    
\usepackage{hyperref}       
\usepackage{url}            
\usepackage{booktabs}       
\usepackage{amsfonts}       
\usepackage{nicefrac}       
\usepackage{microtype}      
\usepackage{xcolor}         
\usepackage{enumitem}

\usepackage[capitalize]{cleveref}

\usepackage[utf8]{inputenc} 
\usepackage[T1]{fontenc}    
\usepackage{hyperref}       
\usepackage{url}            
\usepackage{booktabs}       
\usepackage{amsfonts}       
\usepackage{nicefrac}       
\usepackage{microtype}      
\usepackage{xcolor}         

\usepackage{caption}
\usepackage{subcaption}

\usepackage{wrapfig}

\newcommand{\name}{\texttt{XGNNCert}}

\theoremstyle{plain}

\theoremstyle{definition}
\newtheorem{definition}[theorem]{Definition}

\theoremstyle{remark}

\title{Provably Robust Explainable Graph Neural Networks against Graph Perturbation Attacks}

\author{Antiquus S.~Hippocampus, Natalia Cerebro \& Amelie P. Amygdale \thanks{ Use footnote for providing further information
about author (webpage, alternative address)---\emph{not} for acknowledging
funding agencies.  Funding acknowledgements go at the end of the paper.} \\
Department of Computer Science\\
Cranberry-Lemon University\\
Pittsburgh, PA 15213, USA \\
\texttt{\{hippo,brain,jen\}@cs.cranberry-lemon.edu} \\
\And
Ji Q. Ren \& Yevgeny LeNet \\
Department of Computational Neuroscience \\
University of the Witwatersrand \\
Joburg, South Africa \\
\texttt{\{robot,net\}@wits.ac.za} \\
\AND
Coauthor \\
Affiliation \\
Address \\
\texttt{email}
}

\author{
{\rm Jiate Li$^{1}$,  Meng Pang$^{2}$, Yun Dong$^{3}$, Jinyuan Jia$^{4}$, Binghui Wang$^{1}$}\\
$^1$Illinois Institute of Technology, USA, $^2$Nanchang University, China \\ $^3$Milwaukee School of Engineering, USA  
$^4$The Pennsylvania State University, USA 
}

\iclrfinalcopy 
\begin{document}

\setcounter{theorem}{0}

\maketitle

\input{abstract}

\input{intro}

\input{background}

\input{method}

\input{exp}

\input{related}

\input{conclusion}

\bibliography{iclr2025_conference}
\bibliographystyle{iclr2025_conference}

\input{appendix}

\end{document}

%% file: math_commands.tex
\usepackage{amsmath,amsfonts,bm}

\def\eqref#1{equation~\ref{#1}}

\def\1{\bm{1}}

\DeclareMathAlphabet{\mathsfit}{\encodingdefault}{\sfdefault}{m}{sl}
\SetMathAlphabet{\mathsfit}{bold}{\encodingdefault}{\sfdefault}{bx}{n}



%% file: abstract.tex
\begin{abstract}
Explainable Graph Neural Networks (XGNNs) have garnered increasing attention for enhancing the transparency of Graph Neural Networks (GNNs), which are the leading methods for learning from graph-structured data. While existing XGNNs primarily focus on improving explanation quality, their robustness under adversarial attacks remains largely unexplored. Recent studies have shown that even minor perturbations to graph structure can significantly alter the explanation outcomes of XGNNs, posing serious risks in safety-critical applications such as drug discovery.

In this paper, we take the first step toward addressing this challenge by introducing {\name}, the first provably robust XGNN. {\name} offers formal guarantees that the explanation results will remain consistent, even under worst-case graph perturbation attacks, as long as the number of altered edges is within a bounded limit. Importantly, this robustness is achieved without compromising the original GNN’s predictive performance.   
Evaluation results on multiple graph datasets and GNN explainers show the effectiveness of {\name}. Source code is available at \url{https://github.com/JetRichardLee/XGNNCert}. 
\end{abstract}

%% file: intro.tex
\section{Introduction}
\label{sec:intro}

Explainable Graph Neural Network (XGNN) has emerged recently  to foster the trust of using GNNs---it provides a human-understandable way to interpret the prediction by GNNs. Particularly, given a graph and a predicted node/graph label by a GNN, XGNN aims to uncover the \emph{explanatory edges} (and the connected nodes) from the raw graph that is crucial for predicting the label (see Figure \ref{fig:Explanation}(a) an example). Various XGNN methods~\citep{GNNEx19,DBLP:journals/corr/abs-2011-04573/PGExplainer,DBLP:journals/corr/abs-2102-05152/subgraphX,zhang2022gstarx,wang2023/gnninterpreter,behnam2024graph} have been proposed from different perspectives (more details see Section~\ref{sec:related}) and they have also been widely adopted in applications including  disease diagnosis~\citep{pfeifer2022gnn}, drug analysis~\citep{yang2022mgraphdta,Drug_repurposing2023}, fake news spreader detection~\citep{rath2021scarlet}, and molecular property prediction~\cite{wu2023chemistry}. 
 
While existing works focus on enhancing the explanation performance, the robustness of XGNNs is largely unexplored. 
\citet{li2024graph} observed that  
 well-known XGNN methods (e.g., GNNExplainer~\citep{GNNEx19}, PGExplainer~\citep{DBLP:journals/corr/abs-2011-04573/PGExplainer}) are vulnerable to graph perturbation attacks
---Given a graph, a  GNN model, and a GNN explainer, an adversary can slightly perturb a few edges such that the GNN predictions are accurate, but the explanatory edges outputted by the GNN explainer on the perturbed graph is drastically changed. This attack could cause serious issues in the safety/security-critical applications such as drug analysis. For instance, \cite{Drug_repurposing2023} designs an XGNN tool Drug-Explorer for drug repurposing (reuse existing drugs for new diseases), where users input a drug graph and the tool outputs the visualized explanation result (i.e., important chemical structure) useful for curing the diseases. 
If such tool is misled on adversarial purposes (i.e., adversary inputs a carefully designed perturbed graph), it may recommend invalid drugs with harmful side-effects. Therefore, it is crucial to design defenses for GNN explainers against these attacks.

Generally, defense strategies can be classified as \emph{empirical defense} and \emph{certified defense}. Empirical defenses often can be  broken by stronger/adaptive attacks, as verified in many existing works on defending against adversarial examples~\citep{carlini2019evaluating} and adversarial graphs~\citep{zhang2020backdoor,yang2024distributed}. 
We notice two empirical defense methods  \citep{bajaj2021robust,wang2023vinfor} have been proposed to  robustify XGNNs against graph perturbations. 
Likewise, we found they are ineffective against stronger attacks proposed in \cite{li2024graph} (see Table~\ref{tab:Effectiveness}). 
In this paper, we hence focus on designing certified defense for XGNNs against graph perturbation attacks. {An XGNN is said certifiably robust against a bounded graph perturbation if, for any graph perturbation attack with a perturbation budget that does not exceed this bound, the XGNN consistently produces the same correct explanation (formal definition is in Definition \ref{def:certxgnn}).} 
There are several technical challenges. 
First, GNN explanation and GNN classification are coupled in XGNNs. Robust GNN classifiers do not imply robust GNN explainers, and 
claiming robust explanations without correct GNN classification is meaningless\footnote{Our additional experiments in Figure~\ref{fig:deceive} in Appendix~\ref{app:evaluation} also validate that if the GNN classifier is deceived, the explanation result would be drastically different compared with the groundtruth explanation.}. It is thus necessary to ensure both robust GNN classification and robust GNN explanation. 
Second, there is a fundamental difference to guarantee the robustness of GNN classifiers and GNN explainers. This is because GNN classifiers map a graph to a label, while  GNN explainers map a graph to an edge set. All existing certified defenses~\citep{jia2020certified,jin2020certified, wang2021certified,xia2024gnncert,li2025agnncert} against graph perturbations focus on the robustness of GNN \emph{classifiers} and cannot be applied to robustify GNN explainers.

In this work, we propose {\name}, the first {certifiably robust} XGNN against graph perturbation attacks. Given a testing graph, a GNN classifier, and 
 a GNN explainer,
{\name} consists of three main steps. First, we are inspired by 
existing defenses for  classification~\citep{levine2020randomized, jia2021intrinsic,jia2022certified,xia2024gnncert,yang2024distributed,li2025agnncert} that divide an input (e.g., image) into multiple non-overlapping parts (e.g., patches). 
However, directly applying the idea to divide the test graph into multiple non-overlapping subgraphs does not work well for robustifying GNN explainers. 
One reason is that it is hard for the GNN explainer to determine the groundtruth explanatory edges from each subgraph due to its sparsity. 
To address it, we propose to leverage both the test graph and its complete graph for "hybrid" subgraph generation. An innovation design here is only a bounded number of hybrid subgraphs could be affected when the 
test graph is adversarially perturbed with a bounded perturbation, which is the requirement for deriving the robustness guarantee.  
Second, we build a majority-vote classifier on GNN predictions for the generated hybrid subgraphs,  
and a majority-vote explainer on GNN explanations for interpreting 
the prediction 
of the hybrid subgraphs. 
Last, we derive the certified robustness guarantee. 
Particularly, {\name} guarantees the majority-vote classifier yields the same prediction, and majority-vote explainer outputs close explanatory edges for the perturbed testing graph under arbitrary graph perturbations, when the number of perturbed edges is bounded (which we call \emph{certified perturbation size}). 

We evaluate {\name} on multiple XGNN methods on both synthetic and real-world graph dataset with groundtruth explanations. 
Our results show {\name} does not affect the normal explanation accuracy without attack. Moreover, {\name} shows it can guarantees at least 2
edges are from the 5 groundtruth explanatory edges, when averaged 6.2 edges are arbitrarily perturbed in testing graphs from the SG+House dataset.  
Our major contributions are as follows:
\begin{itemize}[leftmargin=*]
\vspace{-2mm}
    \item We propose {\name}, to our best knowledge, the first certified defense for explainable GNN  against graph perturbation attacks. 
    \item We derive the deterministic robustness guarantee of {\name}. 
    \item We evaluate {\name} on multiple graph datasets and GNN explainers and show its effectiveness.  
\end{itemize}

%% file: background.tex
\section{Background and Problem Formulation}
\label{sec:background}

{\bf GNN and XGNN:} 
 We denote a graph as ${G} = (\mathcal{V}, \mathcal{E})$, that consists of a set of nodes $v \in \mathcal{V}$ and edges $e_{u,v} \in \mathcal{E}$. A GNN, denoted as $f$, takes a graph $G$ as input and outputs a predicted label $y = f(G) \in \mathcal{C}$, with $\mathcal{C}$ including all possible labels. For instance, $y$ can be defined on the graph $G$ in graph classification, or on a specific node $v \in \mathcal{V}$ in node classification. An XGNN, denoted as $g$, uncovers the key component in $G$ that contributes to the GNN prediction $y$. In this paper, we focus on the widely-studied edge explanations where $g$ takes $(G,y)$ as input and determines the important edges in $G$. Particularly, this type of XGNN  learns importance scores $\mathbf{m}$ for all edges $\mathcal{E}$; and selects the edges $\mathcal{E}_k \subseteq \mathcal{E}$ with the top $k$ scores in $\mathbf{m}$ as the explanatory edges, where $k$ is a hyperparameter of the XGNN. Formally, ${\bf m} = g(G,y), \mathcal{E}_k = \mathcal{E}.top_{k}({\bf m}) = \mathcal{E}.top_{k}(g(G,y))$. 

\begin{figure}[!t]
\vspace{-4mm}
    \centering
    \captionsetup[subfloat]{labelsep=none, format=plain, labelformat=empty}

    \subfloat[\small (a) GNN Explanation]{
        \includegraphics[width=0.48\linewidth]{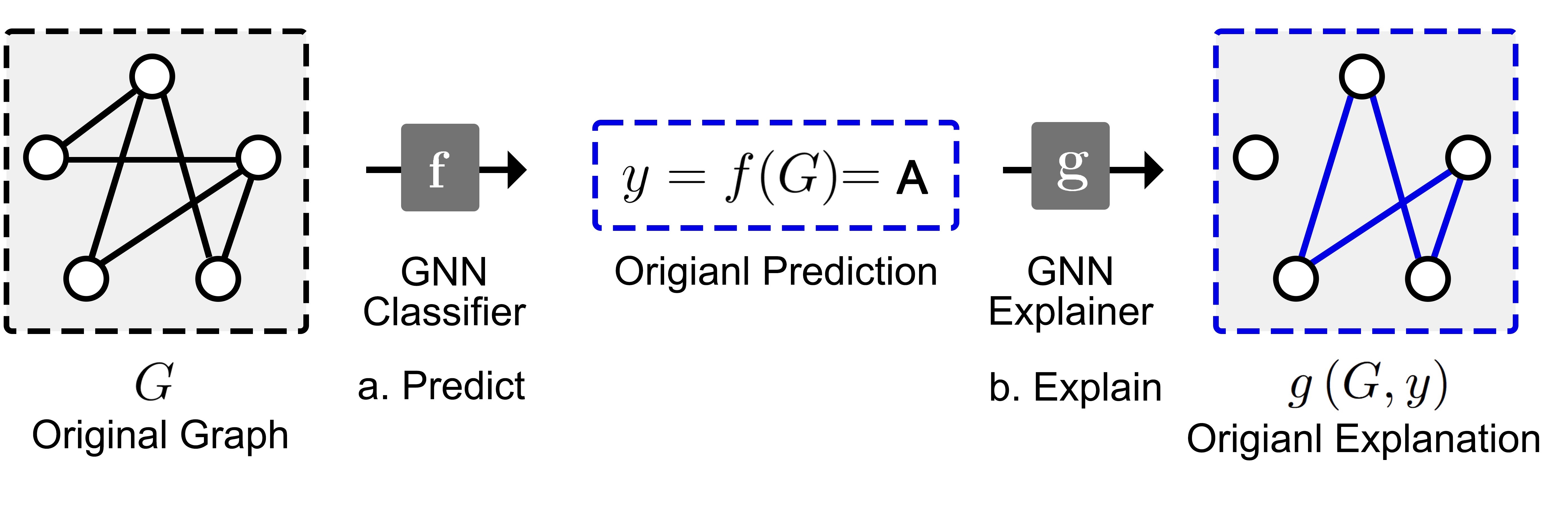}
    }\hfil
    \subfloat[\small (b) Adversarial Attack on GNN Explanation]{
        \includegraphics[width=0.48\linewidth]{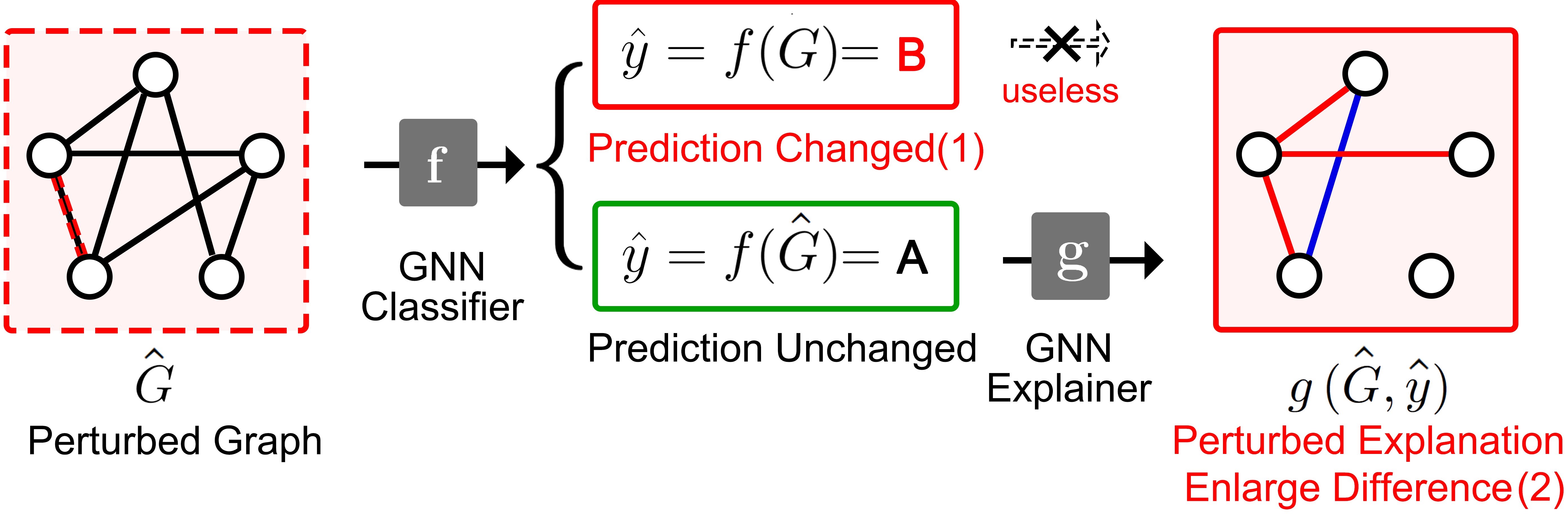}
    }
    \vspace{-2mm}
    \caption{(a) GNN for graph classification and GNN explanation. 
    A GNN classifier $f$ first predicts a label $y$ for the graph $G$, and then a GNN explainer $g$ interprets the predicted label $y$ to produce the explanatory edges $\mathcal{E}_k$. 
    (b) Two possible graph perturbation attacks on the GNN explainer $g$: 1) the GNN prediction $\hat{y}$ on the perturbed graph $\hat{G}$ is different from $y$; 
    2) the GNN prediction on $\hat{G}$ is kept, but the explanatory edges $\hat{\mathcal{E}}_k$ outputted by $g$ after the attack is largely different from $\mathcal{E}_k$.
    }
    \label{fig:Explanation}
    \vspace{-4mm}
\end{figure}

{\bf Adversarial Attack on XGNN:} Given a graph $G$ and a prediction $y$ (by a GNN $f$), an XGNN $g$ and its explanatory edges $\mathcal{E}_k$. We notice that an attacker can perturb the graph structure to mislead the XGNN $g$. To be specific, the attacker could delete edges from $G$ (to ensure stealthiness, the attacker does not delete edges in the explanatory edges   $\mathcal{E}_k$, as otherwise it can be easily identified) or add new edge into $G$. We denote the adversarially perturbed graph as $\hat{G} = (\mathcal{V}, \hat{\mathcal{E}})$, with an attack budget $M$ (i.e., the total number of deleted and added edges is no more than $M$). On the perturbed graph $\hat{G}$, $f$ gives a new prediction $\hat{y}=f(\hat{G})$, and $g$ produces new explanatory edges $\hat{\mathcal{E}}_k = \hat{\mathcal{E}}.top_k(g(\hat{G},\hat{y}))$.

We assume the attacker has two ways to attack an XGNN, as illustrated in Fig~\ref{fig:Explanation} (b). 
\begin{enumerate}[leftmargin=*]
\vspace{-2mm}
\item \emph{The attacker simply misleads the GNN prediction}. Note that if the prediction is changed (i.e., $\hat{y} \neq y$), even $\hat{\mathcal{E}}_k = \mathcal{E}_k$, the explanation is not useful as it explains the wrong prediction. This attack can be achieved via existing evasion attacks on GNNs, e.g., ~\cite{dai2018adversarial,zugner2018adversarial,xu2019topology,mu2021hard,wang2022bandits}.

\item \emph{A more stealthy attack keeps the correct prediction, but largely deviates the explanation result.}  That is, the attacker aims to largely enlarge the difference 
between 
$\mathcal{E}_k$ and $\hat{\mathcal{E}}_k$ with  the prediction unchanged~\citep{li2024graph}. 

\end{enumerate}

{\bf Problem Formulation:} 
The above results show existing XGNNs are vulnerable to effective and stealthy graph perturbation attacks. Also, 
as various works~\citep{carlini2019evaluating,mujkanovic2022defenses} have demonstrated, empirical defenses often can be  broken by advanced/stronger attacks. Such observations and past experiences motivate us 
to design \emph{certifiably robust XGNNs}, i.e., that can defend against the \emph{worst-case} graph perturbation attacks with a bounded attack budget. 

\begin{definition}[$(M_{\lambda},\lambda)$-Certifiably robust XGNN]
\label{def:certxgnn}
We say an XGNN is $(M_{\lambda},\lambda)-$certifiably robust, 
if, for \emph{any} graph perturbation attack with no more than $M_{\lambda}$ perturbed edges on a graph $G$, the GNN prediction on the induced perturbed graph $\hat{G}$ always equals to the prediction $y$ on $G$, and there are at least $\lambda$ ($\leq k$) same edges in the explanatory edges $\hat{\mathcal{E}}_k$ after the attack and the explanatory edges $\mathcal{E}_k$ without the attack. 
We also call $M_{\lambda}$ the \emph{certified perturbation size}. 
{Further, we denote by $M_{\lambda}^*$ the \emph{maximal} $M_{\lambda}$  associated with a $\lambda$, for which a specific XGNN remains certifiably robust.} 
\vspace{-2mm}
\end{definition}

\emph{Remark:} A smaller $\lambda$ implies a larger {$M_\lambda^*$}.  
When $\lambda=k$, a robust XGNN should guarantee $\hat{\mathcal{E}}_k = \mathcal{E}_k$. 
In this paper, we focus on deriving the certifiably robust XGNN for the graph-level classification task. The adaptation of the proposed defense techniques (in Section \ref{sec:xgnncert}) to other graph-related tasks, such as node-level classification and edge-level classification, is discussed in Appendix \ref{app:discussion}.

%% file: method.tex
\section{{\name}: Our Certifiably Robust XGNN}
\label{sec:xgnncert}

\vspace{-2mm}
In this section, we propose  {\name}, our certifiably robust XGNN against graph perturbation attacks.   
Given a testing graph, a GNN classifier, and a GNN explainer, {\name} consists of three major steps. 
\emph{1) Hybrid subgraphs generation:} it aims to generate a set of subgraphs that leverage  the edges from both the testing graph and its  complete graph. 
\emph{2) Majority-voting based classification and explanation:} it builds a majority-vote based classifier (called voting classifier) on GNN predictions for the hybrid subgraphs,
as well as a majority-vote based explainer (called voting explainer) on GNN explanations for interpreting the predicted label of the hybrid subgraphs. 
\emph{3) Deriving the certified robustness guarantee:} based on the generated subgraphs, our voting classifier and voting explainer, it derives the maximum perturbed edges, such that our voting classifier guarantees the same prediction on the perturbed graph and testing graph, and our voting explainer guarantees the explanation results on the perturbed graph and the clean graph are close.
 Figure~\ref{fig:overview} shows an overview of our {\name}. 

\begin{figure}[t]
\vspace{-4mm}
    \centering
    \captionsetup[subfloat]{labelsep=none, format=plain, labelformat=empty}

    \includegraphics[width=\linewidth]{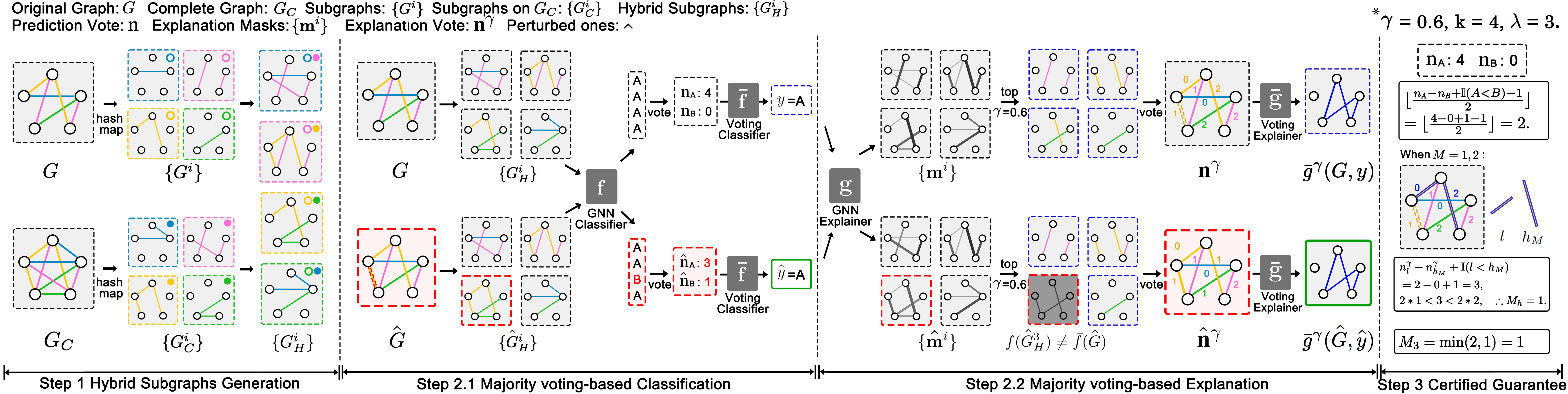}
    \vspace{-4mm}
    \caption{Overview of the proposed three-step certifiably robust XGNN.
    }
    \label{fig:overview}
   \vspace{-2mm}
\end{figure}

\subsection{Hybrid Subgraphs Generation}
\label{sec:Hybrid}
A straightforward idea is to adapt the existing defense strategy for  classification~\citep{levine2020randomized,xiang2021patchguard, jia2021intrinsic,jia2022certified,xia2024gnncert} that divides an input into multiple \emph{non-overlapping} parts. 
Particularly, one can divide a graph into multiple non-overlapping subgraphs, such that only a bounded number of  subgraphs are affected when the graph is adversarially perturbed under a bounded perturbation. 
However, this strategy does not work well to robustify GNN explanation (Our results in Section~\ref{sec:eval} also validate this) due to two reasons: 
(1) Every edge in a graph appears only once in all subgraphs. This  makes it hard for the GNN explainer to ensure the groundtruth explanatory edges to have higher scores than non-explanatory edges. 
(2) All subgraphs only contain existent edges in the graph, while nonexistent edges can be inserted into the graph during the attack and their importance for explanation needs to be also considered. 
To address the challenge, we develop a hybrid subgraph generation method, that consists of two steps shown below.  

{\bf Generating Subgraph Indexes via Hash Mapping:} 
We use the hash function (e.g., MD5) to generate the subgraph indexes {as done in \citet{xia2024gnncert,yang2024distributed}}\footnote{Our theoretical results require the graph division function has two important properties: 1) It is deterministic, such that each edge and node in a graph is deterministically mapped into only one subgraph. This property is the core to derive our theoretical results.  2) It is independent of the graph structure, as otherwise an attacker may reverse-engineer the function to find the relation between the output and input, and possibly break the defense. The used hash function can achieve both properties.}. A hash function takes input as a bit string and outputs an integer (e.g., within a range $[0,2^{128}-1]$). Here, we propose to use the string of edge index as the input to the hash function. For instance, for an edge $e=(u,v)$, we denote its string as $\textrm{str}(u)+\textrm{str}(v)$, where {the $\textrm{str}(\cdot)$ function  transfers the index number into a string in a fixed length (filled with prefix "0"s)}, and``+" means the string concatenation\footnote{For instance, with a 4-bit length, an edge 12-21 is represented as the string "0012" and "0021", respectively. Then the concatenated string between the edge 12-21 is "00120021".}. 
Then we can map each edge using the hash function to a unique index.
Specifically, we denote the hash function as $h$ and assume $T$ groups 
in total. Then for every edge $e=(u,v)$, we compute its subgraph index $i_e$ as\footnote{We put the node with a smaller index (say  $u$) first and let 
$h[\mathrm{str}(v) + \mathrm{str}(u)]=h[\mathrm{str}(u) + \mathrm{str}(v)]$.}. 
\begin{align}
\label{eqn:hashidx}
i_e = h[\mathrm{str}(u) + \mathrm{str}(v)] \, \, \mathrm{mod} \, \, T+1. 
\end{align}

{\bf Generating Hybrid Subgraphs:} Based on the hash function, we can construct a set of $T$ subgraphs for any graph. However, instead of only using existent edges in the given graph to construct subgraphs, we propose to also use \emph{nonexistent edges} to promote the robustness performance for GNN explainers. \emph{A key requirement is: how to guarantee only a bounded number of subgraphs are affected when the original graph is adversarially perturbed.} To address it, we innovatively propose to use the \emph{complete graph},
and our theoretical results in Theorem~\ref{thm:bounddiff} show the requirement can be satisfied.

\emph{Dividing the input graph into subgraphs:} 
For an input graph $G=(\mathcal{V},\mathcal{E})$, we use 
$\mathcal{E}^i$ to denote the set of edges whose subgraph index is $i$, i.e., 
$\mathcal{E}^i = \{\forall e \in \mathcal{E}: i_e= i \}.$ 
Then, we can construct $T$ subgraphs for $G$ as $\{ {G}^i = (\mathcal{V}, \mathcal{E}^i): i=1,2,\cdots, T\}$. 

\emph{Dividing the complete graph into subgraphs:} 
We denote the complete graph of $G$ as $G_C = (\mathcal{V}, \mathcal{E}_C), \mathcal{E}_{C}=\{(u,v), \forall u,v\in \mathcal{V}: u<v\}$. Similarly,  we can divide  $G_C$ into $T$ subgraphs using the same hash function. 
First, the edges having a subgraph index $i$ is denoted as $\mathcal{E}^{i}_{C} = \{\forall e \in \mathcal{E}_{C}: i_{e}=i\}$. Then, we create the $T$ subgraphs for $G_C$ as: $\{{G}_C^i = (\mathcal{V}, \mathcal{E}_C^i): i=1,2,\cdots, T\}$. 

\emph{Hybrid subgraphs:} Now we combine subgraphs $\{G^i\}$ with $\{G_C^i\}$ to construct the hybrid subgraphs. 
For each subgraph $G^i$, we propose to combine it with a fraction (say $p$) of the subgraphs in $\{G_C^i\}$ to generate a hybrid subgraph, denoted as $G_H^i$. There are many ways for the combination, and the only constraint is that the subgraph $G_C^i$ with the same subgraph index $i$ as $G^i$ is not chosen in  $G_H^i$, in order to maintain the information from the original subgraph $G^i$ (otherwise it is overlaid by $G_C^i$). 
Let $\mathcal{T}_{-i} = \mathcal{T}\setminus{i}$ be the index set not including $i$. For instance, we can choose $\lfloor pT\rfloor$ indexes, denoted as $\mathcal{T}_{-i}^p$, from $\mathcal{T}_{-i}$ uniformly at random. 
Then a constructed hybrid subgraph is  $G_H^i = (\mathcal{V}, {\mathcal{E}}_H^{i})$, where 
\begin{align}\label{eq:hybrid}
{\mathcal{E}}_H^{i} =   (\cup_{j \in \mathcal{T}_{-i}^p} \mathcal{E}^{j}_{C})\cup \mathcal{E}^{i}.  
\end{align}
Note that a too small or too large $p$ would degrade the explanation performance. This is because a too large $p$ would make excessive nonexistent edges be added in each $G_H^i$, and a too small $p$ would make explanatory edges be difficult to have higher important scores than non-explanatory edges. Our results show the best performance is often achieved with a modest $p$, e.g., $p \in [0.2,0.4]$.

With the built hybrid subgraphs, we prove in Theorem~\ref{thm:bounddiff} that for any two graphs with $M$ different edges (but same nodes), there are at most $M$ different ones between their respective hybrid subgraphs. 
\emph{We emphasize this is the key property to derive our certified robustness guarantee in Section~\ref{sec:Certify}.}   

\begin{theorem}
[Bounded number of different subgraphs]
\label{thm:bounddiff}
For any two graphs $G=(\mathcal{V},\mathcal{E})$, $\hat{G}=(\mathcal{V}, \hat{\mathcal{E}})$ satisfying $|\mathcal{E}\setminus \hat{\mathcal{E}}| = M$. 
The corresponding hybrid subgraphs generated using the above strategy are denoted as  $\{ G_H^i\}$ and $\{ \hat{G}_H^i\}$, respectively. 
Then $\{ G_H^i\}$ and $\{ \hat{G}_H^i\}$ have at most $M$ different graphs. 
\vspace{-2mm}
\end{theorem}
\begin{proof}
See Appendix~\ref{app:thm:bounddiff}.
\end{proof}

\subsection{Majority Voting-based Classification and Explanation}
\label{sec:Voting}

Inspired by existing works~\citep{levine2020randomized,jia2022certified,xia2024gnncert,yang2024distributed}, we propose to use the majority vote to aggregate the results on the hybrid subgraphs. 
We then design a voting classifier and voting explainer that can respectively act as the robust GNN classifier and robust GNN explainer, as expected. 
Assume we have a testing graph $G$ with label $y$, a set of $T$ hybrid subgraphs $\{G_H^{i}\}$ built from $G$, a GNN classifier $f$, and a GNN explainer $g$.  

{\bf Voting Classifier:} 
We denote by $n_c$ the votes of hybrid subgraphs classified as the label $c$ by $f$, i.e., 
{\begin{align}
\small
\label{eqn:labelcount}
n_c = \sum_{i=1}^{T}\mathbb{I}(f(G_H^{i})=c), \forall c \in \mathcal{C},
\end{align}}
where $\mathbb{I}(\cdot)$ is an indicator function.  Then, we define 
our voting classifier $\bar{f}$\footnote{$\bar{f}$ returns a smaller label index when ties exist.} as: 
\begin{align}
\bar{f}(G) = \underset{c \in \mathcal{C}}{{\arg\max}} \, n_c.
\vspace{-2mm}
\end{align}

{\bf Voting Explainer:} 
Recall that when a GNN explainer interprets the predicted label for a graph, it first learns the importance scores for all edges in this graph and then selects the edges with the highest scores as the explanatory edges. 
Motivated by this, we apply $g$ on the hybrid subgraphs  having the same predicted label as the majority-voted label to obtain the explanatory edges,
and then vote the explanatory edges from these hybrid subgraphs. 
Edges with top-$k$ scores are the final explanatory edges. 
Specifically, for each $G_H^{i}$, we apply $g$ to obtain its edges' importance scores $\textbf{m}^{i}=g(G_H^{i},\bar{f}(G))$.  
We define the votes ${n}_{e}^\gamma$ of each edge $e \in \mathcal{E}_{C}$ as the times that its importance score ${m}_{e}^{i}$ is no less than $\gamma$ fraction of the largest scores in every hybrid subgraph $G_H^{i}$ with the prediction $f(G_H^{i})= \bar{f}(G)$: 
{
\begin{align}
\label{eqn:votes}
{n}_{e}^\gamma = \sum_{i=1}^{T} \mathbb{I}({m}_{e}^{i} \geq {\bf m}_{(\gamma)}^{i}) \cdot \mathbb{I}(f(G_H^{i})= \bar{f}(G)), \forall e \in \mathcal{E}_{C},   
\end{align}
}
where ${\bf x}_{(\gamma)}$ means the $\gamma \cdot \texttt{size}({\bf x})$ largest element in ${\bf x}$ and $\gamma$ is a tuning hyperparameter (we will study its impact in our experiments). We denote  $\textbf{n}^{\gamma}$ as the set of votes for all edges in $\mathcal{E}_C$.  Then we define our voting explainer $\bar{g}^{\gamma}$ as outputting the edges from $G$ with the top-$k$ scores in $\textbf{n}^{\gamma}$\footnote{When two edges have the same $\textbf{n}^{\gamma}$, the edge with a smaller index is selected by $\bar{g}^{\gamma}$.}:
\begin{align}
\bar{g}^{\gamma}(G,\bar{f}(G)) = \mathcal{E}.top_{k}(\textbf{n}^{\gamma}).   
\end{align}

{ 
\emph{Remark:}
Traditional GNN classifiers are designed to be node permutation invariant \citep{kipf2017semi,velivckovic2018graph,xu2018powerful}, meaning that the model’s predictions remain consistent regardless of how the nodes in the graph are permuted.
In contrast, our voting classifier is node permutation variant due to the properties of the hash function. This implies that both the classification and explanation performances of {\name} may vary depending on the node ordering. However, we empirically observed that {\name}'s performance remains relatively stable across different node permutations (see Table \ref{tab:gnncert_runs} in Appendix \ref{app:discussion}). 
Moreover, recent studies \citep{Loukas2020What,papp2021dropgnn,huang2022going} suggest that node-order sensitivity can actually enhance the expressivity and generalization capabilities of GNNs. Additional discussions are provided in Appendix \ref{app:discussion}.
}

\subsection{Certified Robustness Guarantee}
\label{sec:Certify}
In this section, we derive the certified robustness guarantee against graph perturbation attacks using our graph division strategy and introduced robust voting classifier and voting explainer. 

We first define some notations. We let $y = \bar{f}(G)$ by assuming the voting classifier $\bar{f}$ has an accurate label prediction, and $\mathcal{E}_k = \bar{g}^{\gamma}(G,y)$ by assuming the voting explainer $\bar{g}$  has an accurate explanation.  
We denote $\bar{G} = (\mathcal{V}, \bar{\mathcal{E}})$ as the complement of $G$, and $\bar{\mathcal{E}}_M$ the edges  in $\bar{\mathcal{E}}$ with top-$M$ scores in $\textbf{n}^{\gamma}$. 
{We introduce the non-existent edges $\bar{\mathcal{E}}_M$ with top scores by considering that, in the worst-case attack with $M$ edge perturbations, $\bar{\mathcal{E}}_{M}$ would be chosen to compete with the true explanatory edges.}

\begin{theorem}
[Certified Perturbation Size $M_\lambda$ for a given $\lambda$] 
\label{thm:VerifyExp} 
Assume $y \in \mathcal{C}$ and ${b} \in \mathcal{C} \setminus \{y\}$ be the class with the highest votes $n_y$ and second highest votes by Eqn (\ref{eqn:labelcount}), respectively. 
 Assume further the edge $l \in \mathcal{E}_k$ is with the $\lambda$-th highest votes $n_{l}^{\gamma}$, and edge $h_M \in \bar{\mathcal{E}}_M \cup (\mathcal{E} \setminus \mathcal{E}_k)$ with the $(k-\lambda+1)$-th highest votes ${n}_{h_M}^{\gamma}$ in ${\textbf{n}}^{\gamma}$ by Eqn (\ref{eqn:votes}) ($h_M$ hence depends on $M$). Then $M_{\lambda}$ satisfies: 
\begin{align}
\label{eqn:CPS}
& M_\lambda \leq M^*= \min \big( \lfloor \frac{n_y-n_{b} + \mathbb{I}(y<b)-1}{2} \rfloor, M_h), \textrm{ where } \\ 
& M_h = \max \, M,  \quad {s.t.} \quad {n}_{l}^\gamma- {n}_{h_M}^\gamma + \mathbb{I}(l<h_M) > 2M. \label{Eq:ExpCondition}
\end{align}

\vspace{-4mm}
\end{theorem}
\begin{proof}
See Appendix \ref{app:thm:VerifyExp}. 
\vspace{-2mm}
\end{proof}

\emph{Remark:} We have the following remarks from Theorem~\ref{thm:VerifyExp}:
\begin{itemize}[leftmargin=*]
\vspace{-2mm}
\item Our voting classifier and voting explainer can tolerate  $M^*$\footnote{{In general, $M^* \leq M^*_\lambda$ in Def.~\ref{def:certxgnn}. 
We will leave it as future work to prove whether the derived $M^*$ is tight. 
}} perturbed edges. 

\item Our voting classifier can be applied for \emph{any} GNN classifier and our voting explainer for any GNN explainer that outputs edge importance score.

\item Our certified robustness guarantee is deterministic, i.e., it is true with a probability of 1.  

\end{itemize}

%% file: exp.tex
\section{Evaluation}
\label{sec:eval}
\vspace{-2mm}

\subsection{Experimental Setup}
\vspace{-2mm}

{\bf Datasets:} 
As suggested by~\citep{agarwal2023evaluating}, we choose datasets with groundtruth explanations for evaluation. 
We adopt the synthetic dataset "SG-Motif", where each graph has a label and ``Motif" is the groundtruth explanation that can be "House", "Diamond", and "Wheel".
We also adopt two real-world graph datasets (i.e., Benzene and FC) with groundtruth explanations from~\cite{agarwal2023evaluating}. 
Their dataset statistics are described in Table~\ref{tab:Datasets} in Appendix~\ref{app:evaluation}.
For each dataset, we randomly sample 70\% graphs for training, 10\% for validation, and use the remaining 20\% graphs for testing. 

{\bf GNN Explainer and Classifier:}
Recent works \citep{funke2022zorro,agarwal2023evaluating} show many GNN explainers (including the well-known GNNExplainer~\cite{GNNEx19}) are unstable, i.e., they yield significantly different explanation results under different runs. We also validate this and show results in Table~\ref{tab:instable} in Appendix. This makes it hard to evaluate the explanation results in a consistent or predictable way. To avoid the issue, we carefully select XGNN baselines with stable explanations: PGExplainer~\citep{DBLP:journals/corr/abs-2011-04573/PGExplainer}, Refine~\citep{wang2021towards}, and GSAT~\citep{DBLP:journals/corr/abs-2201-12987/GSAT}. We also select three well-known GNNs as the GNN classifier: GCN~\citep{kipf2017semi}, GSAGE~\citep{hamilton2017inductive}, and GIN~\citep{xu2018how}.  
We implement  these explainers and classifiers using their publicly available source code. 
Appendix~\ref{app:evaluation} details our training strategy to learn the voting explainer and voting classifier in {\name}. 

{\bf Evaluation Metrics:} We adopt three metrics for evaluation.  
\emph{1) Classification Accuracy}: fraction of testing graphs that are correctly classified, e.g., by our voting classifier; \emph{2) Explanation Accuracy}: fraction of explanatory edges outputted, e.g., by our voting  explainer, are in the groundtruth across all testing graphs; 
\emph{3) Certified Perturbation Size  $M^*$ at Certified Explanation Accuracy (or $\lambda$)}: 
Given a testing graph with groundtruth ($k$) explanatory edges, and a predefined $\lambda$ (or certified explanation accuracy $\lambda/k$), our theoretical result outputs (at least) $\lambda$ explanatory edges on the perturbed testing graph are from the groundtruth, where the testing graph allows arbitrary {$M^*$} perturbations.  
{$M^*$ vs  $\lambda$} then reports the average {$M^*$} of all testing graphs for the given $\lambda$. 

{\bf Parameter Setting:} 
There are several hyperparameters in our {\name}. Unless otherwise mentioned, we use GCN as the default GNN classifier and PGExplainer as the default GNN explainer. 
we use MD5 as the hash function $h$ and we set $\lambda = 3$, $p=0.3$, $T = 70$, $\gamma = 0.3$ and $k$ as Table~\ref{tab:Datasets}. We will also study the impact of these hyperparameters on our defense performance. 

\begin{table}[!t]
    \centering
    \renewcommand\arraystretch{1.15}
    \resizebox{\textwidth}{!}{
    \begin{tabular}{c|c|cccc|c|cccc|c|cccc}
		\toprule
      \multirow{2}*{\bf Datasets} & \multicolumn{5}{c|}{\bf PGExplainer} & \multicolumn{5}{c|}{\bf ReFine} &  \multicolumn{5}{c}{\bf GSAT}\\
         \cline{2-16}
         &\multirow{2}{*}{Orig.}&\multicolumn{4}{c|}{T}&\multirow{2}{*}{Orig.}&\multicolumn{4}{c|}{T}&\multirow{2}{*}{Orig.}&\multicolumn{4}{c}{T}\\
        \cline{3-6} \cline{8-11} \cline{13-16}
         & &30&50&70&90&  &30&50&70&90& &30&50&70&90\\
         \Xhline{0.7px}
         SG+House&0.740&0.658&0.725&{0.725}&0.673&0.707&0.588& {0.690}& 0.593&0.564&0.744& {0.759}&0.716&0.673&0.658\\
         \cline{1-16}
        SG+Diamond& 0.745&0.704& {0.730}&0.729& 0.620&0.569&0.440&0.499& 0.521&0.398&0.564&0.426&0.493&0.558&0.420\\
          \cline{1-16}
        SG+Wheel&0.629&0.587&0.612&0.571& 0.542&0.604&0.614& {0.626}& 0.606&0.462&0.568&0.491&0.544&0.612&0.562\\
         \cline{1-16}
         {Benzene}&0.552&0.421&0.497&0.468& 0.429&0.559&0.463& 0.474&  {0.512}&0.440&0.552&0.314&0.430&0.445&0.398\\
         \cline{1-16}
        {FC}&0.531&0.385&0.452&0.373& 0.328&0.503&0.369&0.447& 0.425&0.314&0.487&0.350&0.392&0.412& 0.373\\
        \bottomrule
    \end{tabular}}
    \vspace{-2mm}
    \caption{Explanation accuracy on the original GNN explainers and our {\name}.}
   \vspace{-2mm}
   \label{tab:Explanation}
\end{table}

\begin{table}[!t]
    \footnotesize
        \centering 
    \renewcommand\arraystretch{0.9}
     \resizebox{\textwidth}{!}{\begin{tabular}{c|c|cccc|c|cccc|c|cccc}
     \toprule
       \multirow{2}{*}{\bf Datasets}   & 
       \multicolumn{5}{c|}{\bf GCN}&  \multicolumn{5}{c|}{\bf GIN}&  \multicolumn{5}{c}{\bf GSAGE} \\
         \cline{2-16} 
        & \multirow{2}{*}{Orig.} & \multicolumn{4}{c|}{T}&\multirow{2}{*}{Orig.} & \multicolumn{4}{c|}{T}&\multirow{2}{*}{Orig.} & \multicolumn{4}{c}{T}\\
        \cline{3-6} \cline{8-11} \cline{13-16}
         & &30&50&70&90& &30&50&70&90& &30&50&70&90\\
         \Xhline{0.7px}
         SG+House& 0.920&0.895&0.905&0.905&0.890&0.945&0.915&0.915&0.900&0.905&0.930&0.900&0.890&0.895&0.875\\
         \cline{1-16}
        SG+Diamond& 0.965&0.935&0.935&0.935&0.930& 0.975&0.935&0.955&0.955&0.955& 0.965&0.940&0.940&0.940&0.940\\
         \cline{1-16}
        SG+Wheel& 0.915
&0.905&0.905&0.900&0.885& 0.930
&0.915&0.905&0.900&0.895& 0.920
&0.910&0.910&0.895&0.890\\
         \cline{1-16}
         Benzene& 0.758&0.746&0.700&0.723&0.707& 0.792&0.736&0.754&0.754&0.754& 0.773&0.725&0.760&0.718&0.718\\
         \cline{1-16}
         FC&0.711&0.674&0.692&0.692&0.631&0.800&0.662&0.714&0.714&0.703&0.723&0.692&0.692&0.692&0.620\\
       \bottomrule
    \end{tabular}}
    \vspace{-2mm}
    \caption{Prediction accuracy on the original GNN classifiers and our {\name}.}
    \vspace{-2mm}
    \label{tab:Prediction}
\end{table}

\begin{table}[!t]
    \tiny
\begin{minipage}[c]{\textwidth}
    \centering
    \resizebox{\textwidth}{!}{
        \begin{tabular}{c|cccc|ccc|ccc}
		\toprule
      \multirow{2}*{\bf Datasets}  &\multicolumn{4}{c|}{$p$}&\multicolumn{3}{c|}{$\gamma$}&\multicolumn{3}{c}{$h$}\\
         \cline{2-11}
         &0.0&0.2&0.3&0.4&0.2&0.3&0.4&{\fontsize{5}{5}\selectfont MD5}&{\fontsize{5}{5}\selectfont SHA1}&{\fontsize{4.5}{5}\selectfont SHA256}\\
         \Xhline{0.7px}
         SG+House&0.053&0.695&0.725&0.710& 0.715&0.725&0.720& 0.725&0.718&0.710\\
         \cline{1-11}
        SG+Diamond& 0.045&0.620&0.729&0.720& 0.712& 0.729&0.718&0.729&0.729&0.721\\
          \cline{1-11}
        SG+Wheel& 0.042&0.511&0.571&0.508& 0.550&0.571&0.564& 0.571&0.565&0.562\\
         \cline{1-11}
         {Benzene}&0.102&0.433&0.468&0.403& 0.440& 0.468&0.452&0.468&0.472&0.468\\
         \cline{1-11}
         {FC}&0.096&0.353&0.373&0.288& 0.345&0.373&0.385& 0.373& 0.382&0.390\\
        \bottomrule
    \end{tabular}
    }
   \vspace{-2mm}
    \caption{Explanation accuracy of our {\name} under different $p$, $\gamma$, and the hash function $h$}
    \label{tab:ablation}
\end{minipage}
\vspace{-4mm}
\end{table}

\subsection{Evaluation Results}
We first show the explanation accuracy and classification accuracy of {\name} under no attack, to validate it can behave similarly to the conventional GNN classifier and GNN explainer. We then show the guaranteed robustness performance of our {\name} against the graph perturbation attack.

\subsubsection{Explanation Accuracy and Classification Accuracy}
\label{sec:res_emp}
{\bf {\name} maintains the explanation accuracy and classification accuracy on the original GNN explainers  and GNN classifiers:}
Table~\ref{tab:Explanation} shows the explanation accuracy of our {\name} and the original GNN explainers for reference.  
We can observe that {\name} can achieve close explanation accuracies (with a suitable number of subgraph $T$) as the original GNN explainers (which have   different explanation accuracies, due to their different explanation mechanisms). 
This shows the potential of {\name}
as an ensemble based XGNN.   
We also show the classification performance of our voting classifier in {\name} in Table~\ref{tab:Prediction} and  
the original GNNs  classifier for reference. Similarly, we can see our voting classifier learnt based on our training strategy can reach close classification accuracy as the original GNN classifiers.

{\bf Impact of hyperparameters in {\name}:} Next, we will explore the impact of important hyperparameters that  could affect the performance of {\name}. 

\emph{Impact of $T$:} 
Table~\ref{tab:Explanation} shows the 
explanation accuracy of {\name} with different $T$. 
We can see the performance depends on $T$ and the best $T$ in different datasets is different (often not the largest or smallest $T$). Note that the generated hybrid subgraphs use nonexistent edges from the complete graph. If $T$ is too small,  a hybrid subgraph   contains more nonexistent edges, which could exceed the tolerance of the voting explainer. In contrast, a too large $T$ yields very sparse subgraphs, making the useful information in the subgraph that can be used by the explainer be insufficient. This thus makes it hard to ensure explanatory edges have higher important scores than non-explanatory edges. 

\emph{Impact of $p$:} Table~\ref{tab:ablation} shows the explanation accuracy with different $p$, the fraction of the subgraphs generated by the complete graph that are combined with the clean graphs' subgraphs. 
We have a similar observation that a too small or too large $p$ would degrade the explanation performance, with $p=0.3$ obtaining the best performance overall. 
Note that when $p=0$, we only use the information of the original graphs, and the explanation performance is extremely bad. That means it is almost impossible to obtain the groundtruth explanatory edges. 
This thus inspires us to reasonably leverage extra information not in the original graph to guide finding the groundtruth explanatory edges. 

\emph{Impact of $\gamma$:} Table~\ref{tab:ablation} also shows the explanation accuracy with different $\gamma$, the fraction of the edges with the largest scores used for the voting explainer. We  can observe the results are relatively stable in the range $\gamma=[0.2,0.4]$. This is possibly due to that important edges in the original graph are mostly within these edges with the largest scores. 

\emph{Impact of $h$:} The explanation accuracy with different hash functions $h$ are shown in Table~\ref{tab:ablation}. We see the results are insensitive to $h$, suggesting we can simply choose the most efficient one in practice. 

\begin{figure*}[!t]
	\centering
 \subfloat[{SG+House}]
	{\centering\includegraphics[scale=0.225]{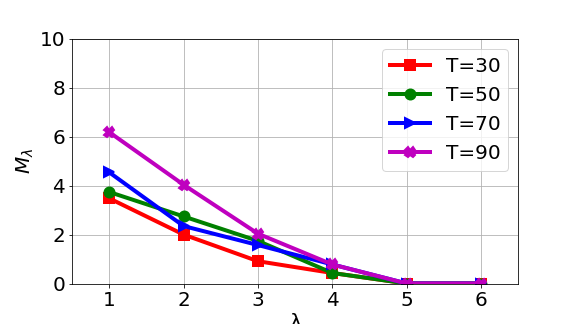}}
 \,
	\subfloat[{SG+Diamond}]
	{\centering \includegraphics[scale=0.22]{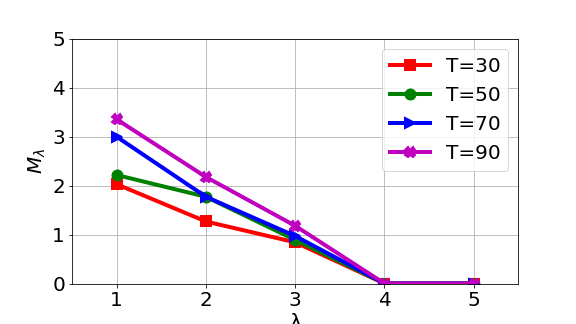}}
\,
\subfloat[{Beneze}]
	{\centering\includegraphics[scale=0.22]{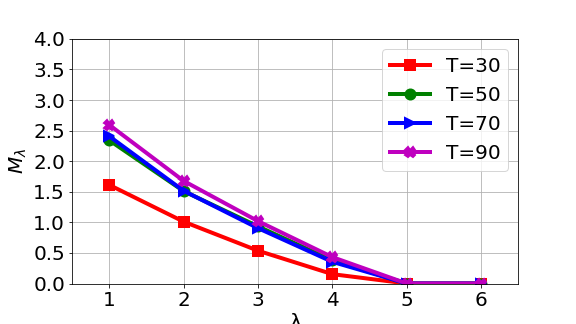}}
	\caption{Certified perturbation size over all testing graphs vs. $\lambda$ on PGExplainer. The maximum $\lambda$ in x-axis equals to $k$, the number of edges in the groundtruth explanation.} 
	\label{fig:CEA_T_pge}
	 \vspace{-8mm}
\end{figure*}

\begin{figure*}[!t]
	\centering
 \subfloat[{SG+House}]
	{\centering\includegraphics[scale=0.22]{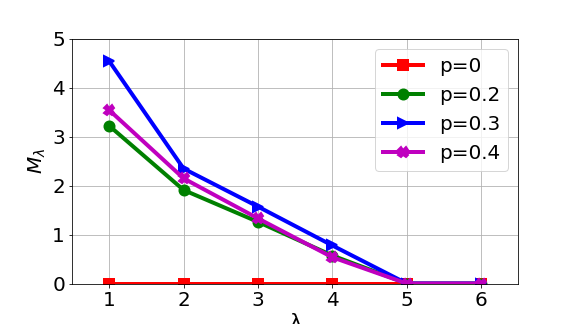}}
 \,
	\subfloat[{SG+Diamond}]
	{\centering \includegraphics[scale=0.22]{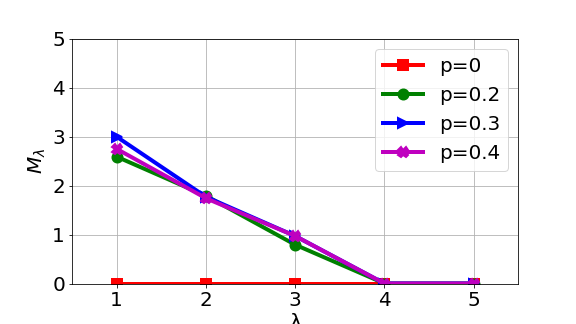}}
\,
 \subfloat[{Beneze}]
 	{\centering\includegraphics[scale=0.22]{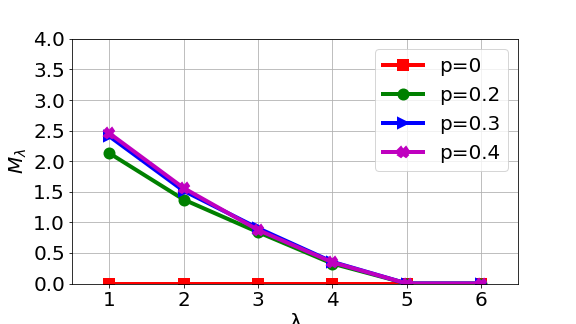}}
	\caption{Certified perturbation size over all testing graphs  vs. $p$ on PGExplainer.} 
	\vspace{-8mm}
	\label{fig:CEA_p_PG}
\end{figure*}

\begin{figure*}[!t]
	\centering
 \subfloat[{SG+House}]
	{\centering\includegraphics[scale=0.22]{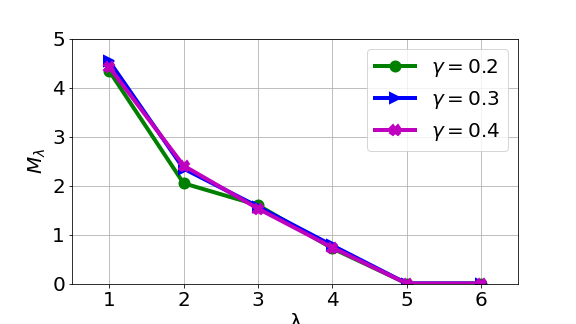}}
 \,
	\subfloat[{SG+Diamond}]
	{\centering \includegraphics[scale=0.22]{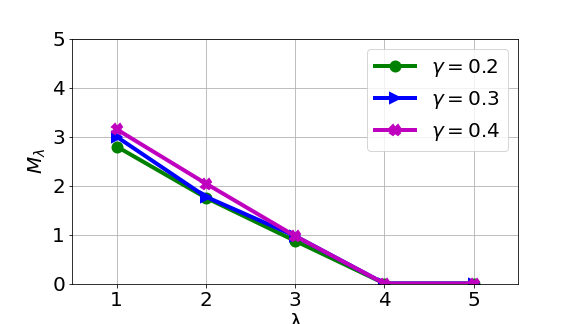}}
\,
 \subfloat[{Beneze}]
 	{\centering\includegraphics[scale=0.22]{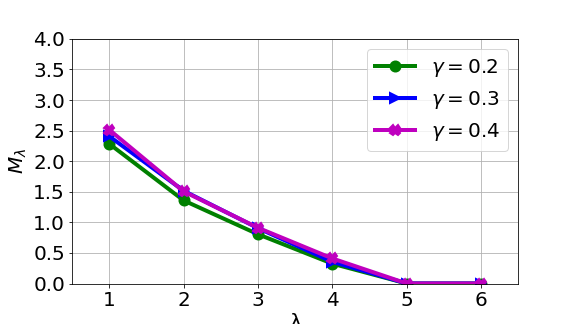}}
	\caption{Certified perturbation size over all testing graphs vs. $\gamma$ on PGExplainer.} 
	\label{fig:CEA_gamma_PG}
	 \vspace{-4mm}
\end{figure*}

\subsubsection{Certified Explanation Accuracy vs. Certified Perturbation Size}
The certified robustness results are shown in Figures~\ref{fig:CEA_T_pge}-\ref{fig:CEA_h_PG}. Due to limited space, we only show results on three datasets and put results on the other datasets and impact of hyperparameters in Appendix~\ref{app:evaluation}. 

{\bf Impact of $T$:} 
Figure~\ref{fig:CEA_T_pge} and Figures~\ref{fig:CEA_T_pge_2}-\ref{fig:CEA_T_gsat_2} in Appendix~\ref{app:moreresults} show the (average) maximum certified perturbation size vs $\lambda$ with different $T$. 
First, XGNNCert obtains reasonable certified explanation accuracy ($\lambda / k$) against the worst-case graph perturbation, when the \#perturbed edges is bounded by ${M^*}$. 
E.g., with average 6.2 edges are arbitrarily perturbed in SG+House, XGNNCert guarantees at least $\lambda=2$ edges are from the $k=5$ groundtruth explanatory edges. 
Second, there exists a tradeoff between the clean explanation accuracy and robust explanation accuracy. Specifically, as $T$ grows, the derived 
 certified perturbation size increases in general. This means that a larger number of generated subgraphs can enlarge the gap between the largest and  second-largest votes in ${\bf n}^{\gamma}$. On the other hand, the explanation accuracy (under no attack) can be decreased as shown in Section~\ref{sec:res_emp}.

{\bf Impact of $p$:} The results are in Figure~\ref{fig:CEA_p_PG} and Figure~\ref{fig:CEA_p_PG-2} in Appendix~\ref{app:moreresults}. First, we  observe the certified perturbation size is 0 when $p=0$. This means, without using information in the complete graph, it is impossible to provably defend against the graph perturbation attack.  Second, the certified explanation accuracies are close when $p$ is within the range $[0.2,0.4]$ (which is different from the conclusions on explanation accuracy without attack). This implies, for each clean graph, we can use 20\%-40\% of the subgraphs generated by the complete graph for achieving stable certified explanation accuracy.    

{\bf Impact of $\gamma$:} The results are shown in Figure~\ref{fig:CEA_gamma_PG} and Figure~\ref{fig:CEA_gamma_PG-2} in Appendix~\ref{app:moreresults}. 
Similarly, the certified results are relatively stable 
in the range $\gamma = [0.2, 0.4]$. The key reason could be that important edges in the original graph are mostly within the edges in these range. 

{\bf Impact of $h$:} The results are shown in Figure~\ref{fig:CEA_h_PG}. Like results on explanation accuracy, we can see the hash function $h$ almost does not affect the certified explanation accuracy. Again, this suggests we can choose the most efficient one in practice. 

\begin{table}[!t]
    \centering
    \small
    \begin{tabular}{c|c|c|c|c|c|c}
         \Xhline{0.8pt} 
      \multicolumn{2}{c|}{\bf Datasets} & {SG+House}&SG+Diamond&SG+Wheel&Benzene&FC\\
         \Xhline{0.8pt} 
        \multirow{2}{*}{\bf Exp. Acc.}&{\bf V-InfoR}&0.693&0.419 &0.439 &0.345 &0.217\\
         \cline{2-7}
         &{\bf {\name}}&0.740&0.729&0.571& 0.468&0.403 \\
         
         \Xhline{0.8pt} 
        {\bf Difference}
        &{\bf V-InfoR}&48.39\% & 73.07\% & 65.35\% &83.82\% &63.22\%\\
         \cline{2-7}
         
         {\bf Fraction}&{\bf {\name}}&7.44\% & 0.0\% & 4.20\% & 1.44\% & 1.41\%\\
         \Xhline{0.8pt} 
    \end{tabular}
    \caption{Explanation accuracy and the fraction of different edges under attack in \citet{li2024graph}.}
    \label{tab:Effectiveness}
    \vspace{-4mm}
\end{table}

{
\subsubsection{Defense Effectiveness Against Adversarial Attack on XGNN}
We further test {\name} in the default setting against the recent adversarial attack on XGNN~\citep{li2024graph}, and compare with the  state-of-the-art empirical defense V-InfoR~\citep{wang2023vinfor}. We evaluate their effectiveness by allowing the attacker to change two non-explanatory edges in the graph and taking the fraction of different explanatory edges (before and after the attack) as the metric. The test results are shown in Table~\ref{tab:Effectiveness}. We can observe that: Our {\name} not only achieves the theoretical defense performance and higher explanation accuracy, but also shows much better empirical defense performance than V-InfoR under the powerful attack. This is possibly due to our subgraph division and voting scheme design, which is ``inherently" robust to the strongest attack---it dilutes the adversarial perturbation effect into subgraphs, and at the same time, the number of subgraphs that are affected can be bounded. In contrast, V-InfoR is an empirical defense that constrains the attack capability and is unable to defend against the strong attack. 

\subsubsection{Complexity Analysis of {\name}}
Our {\name} divides each hybrid graph into $T$ subgraphs and applies a base GNN explainer to explain each subgraph. The final explanation is obtained via voting the explanation results of the $T$ subgraphs, whose computational complexity is negligible. Hence, the dominant computational complexity of {\name} is $T$ times of the base GNN explainer’s. For instance, PGExplainer has a complexity of $O(S|V|+|E|)$, where $S$ is the number of optimization steps, and $|V|$ and $|E|$ are the number of nodes and edges, respectively. Therefore, {\name} with PGExplainer as the base explainer has complexity $O(TS|V|+|E|)$. 
Note that the explanation on $T$ subgraphs can be run in parallel, as they are independent of each other. 
{Furthermore, each hybrid subgraph needs to store $p|V|^2$ more edges from the complete graph, where an edge is represented as a pair of node indexes in the implementation. Hence, the extra memory cost per graph is $O(pT|V|^2))$.}
We highlight that the extra computation and memory cost is to ensure the robustness guarantee. In other words, our {\name} obtains a robustness-efficiency tradeoff.
}

%% file: related.tex
\section{Related Work}
\label{sec:related}

\vspace{-2mm}
{\bf Explainable GNNs:} XNNGs can be classified into \textit{decomposition-based}, \textit{gradient-based},  \textit{surrogate-based}, \textit{generation-based},  \textit{perturbation-based}, and \textit{causality-based} methods. \textit{Decomposition-based methods}~\citep{schnake2021higher,feng2023degree} treat the GNN prediction  as a score and decompose it backward layer-by-layer until reaching the input. The score of different parts of the input indicates the importance to the prediction. 
\textit{Gradient-based methods}~\citep{baldassarre2019explainability,pope2019explainability} 
take the gradient (implies sensitivity) of the prediction wrt. the input graph, and the sensitivity is used to explain the graph for that prediction.
\textit{Surrogate-based methods}~\citep{vu2020pgm,pereira2023distill} replace the complex GNN model with a simple and interpretable surrogate model.
\textit{Generation-based methods} 
~\citep{GEM,sui2022causal,shan2021reinforcement/RGExplainer,Wang_2023/RCExplainer}
use generative models to generate explanations. 
E.g., RCExplainer~\citep{Wang_2023/RCExplainer} 
applies reinforcement learning to generate subgraphs as explanations. 
\textit{Perturbation-based methods}~\citep{GNNEx19,DBLP:journals/corr/abs-2011-04573/PGExplainer,wang2021towards,funke2022zorro}
uncover the important subgraph as explanations by perturbing the input graph. 
\emph{Causality-based methods} \citep{behnam2024graph} explicitly build the structural causal model for a graph, based on the common assumption that a graph often consists of a underlying causal  subgraph. It then adopts the trainable neural causal model \citep{xia2021causal} to learn the cause-effect among nodes for causal explanation.

{\bf Adversarial attacks on GNN classifiers and explainers:}
Almost all existing method  focus on attacking GNN classifiers. 
 They are classified as test-time attacks
\citep{dai2018adversarial,zugner2018adversarial,ma2020towards,mu2021hard,wang2022bandits,wang2023turning,wang2024efficient} and training-time attacks
\citep{xu2019topology,zugner2019adversarial,wang2019attacking,zhang2021backdoor,wang2023turning}. 
Test-time attacks carefully perturb test graphs  
so that as many as them are misclassified by a pretrained GNN classifier, 
while training-time attacks carefully perturb training graphs during training, such that the learnt GNN classifier mispredicts as many test graphs as possible. 
\citep{li2024graph} is the only method on attacking GNN explainers.  
It is a black-box attack (i.e., attacker has no knowledge about XGNN) that aims to corrupt GNN explanations while preserving GNN predictions.

{\bf Certified defenses for GNN classifiers with probabilistic guarantees:} Existing certified defenses~\citep{bojchevski2020efficient,wang2021certified,zhang2021backdoor} are for GNN classifiers--they guarantee the same predicted label for a testing graph with arbitrary graph perturbation. 
For instance,
\citet{wang2021certified}  generalized randomized smoothing~\citep{lecuyer2019certified,cohen2019certified,hong2022unicr} from the continuous domain to the discrete graph domain.~\citet{zhang2021backdoor} extended randomized ablation~\citep{levine2020robustness} to build provably robust graph classifiers. 
However, these defenses only provide probabilistic guarantees and cannot be applied to XGNNs.

{\bf Majority voting-based certified defenses with deterministic guarantees:} 
This strategy has been widely used for classification models against adversarial tabular data~\citep{hammoudeh2023feature},
adversarial 3D points~\citep{zhang2023pointcert}, adversarial patches~\citep{levine2020randomized,xiang2021patchguard}, adversarial graphs~\citep{xia2024gnncert,yang2024distributed,li2025agnncert}, and {data poisoning attacks~\citep{levine2020deep,jia2021intrinsic,wang2022improved,jia2022certified}}. 
Their key differences are creating problem-specific voters for majority voting. 
 {However, these defenses  cannot be applied to 
robustify GNN explainers, which are drastically different from classification models.} 

{\bf Certified defenses of explainable non-graph models.} A few works~\citep{levine2019certifiably,liu2022certifiably,tan2023robust} propose to provably robustify explainable non-graph (image) models against adversarial perturbations. These methods mainly leverage the idea of randomized smoothing \citep{lecuyer2019certified,cohen2019certified} and only provide probabilistic certificates.

%% file: conclusion.tex
\vspace{-2mm}
\section{Conclusion}
\label{sec:conclusion}

\vspace{-2mm}
We propose the first provably robust XGNN (\emph{\name}) against graph perturbation attacks. {\name} first generates multiple hybrid subgraphs for a given graph (via hash mapping) such that only a bounded number of these subgraphs can be affected when the graph is adversarially perturbed. 
We then build a robust voting classifier and a robust voting explainer to aggregate the prediction and explanation results on the hybrid subgraphs. 
Finally, we can derive the robustness guarantee based on the built voting classifier and voting explainer against worst-case graph perturbation attacks  with bounded perturbations. Experimental results on multiple datasets and GNN classifiers/explainers validate the effectiveness of our {\name}.   
{In future work, we will enhance the certified robustness with better subgraph generation strategies and design node permutation invariant certified defenses.}

\subsubsection*{Acknowledgments}
We thank the anonymous reviewers for their insightful reviews. This work was partially supported by the Cisco Research Award and the National Science Foundation under grant Nos. ECCS-2216926, CCF-2331302, CNS-2241713 and CNS-2339686.

%% file: appendix.tex
\appendix

\section{Proofs}

\subsection{Proof of Theorem~\ref{thm:bounddiff}}
\label{app:thm:bounddiff}

When an edge $e$ is added to or deleted from $G$, only the subgraph $G^{i_e} = (\mathcal{V}, \mathcal{E}_{i_{e}})$ is corrupted after hash mapping, and all the other subgraphs $\{G^{j}\}_{j\neq i_e}$ are unaffected. Note that the complete graph $G_C$ is fixed and all subgraphs built from it are never affected. 
Then, with Equation (\ref{eq:hybrid}), only the hybrid subgraph $G_H^{i_{e}}$ would be corrupted. 
Further, when $M$ edges from $G$ are perturbed to form $\hat{G}$, only hybrid subgraphs containing these edges would be corrupted. As some edges may be mapped to the same group index, the different subgraphs between $\{ G_H^i\}$ and $\{ \hat{G}_H^i\}$ is at most $M$.

\subsection{Proof of Theorem~\ref{thm:VerifyExp}}
\label{app:thm:VerifyExp}

After the graph perturbation, we want to satisfy two requirements:   
(1) the voting classifier still predicts 
the class $y$ in the perturbed graph $\hat{G}$  with more votes than predicting any other class $\forall b\in \mathcal{C}\setminus\{y\}$; (2)
the voting explainer ensures at least $\lambda$ edges in $\mathcal{E}_k$ are still in $\hat{\mathcal{E}}_k$, or at most $(k-\lambda)$ edges in ${\mathcal{E}}_C \setminus\mathcal{E}_k$ 
have higher votes than the minimum votes of the edges in ${\mathcal{E}}_k$.

We first achieve (1):  Based on Theorem~\ref{thm:bounddiff}, with any $M$ perturbations on a graph $G$, at most $M$ hybrid subgraphs from $\{G_H^i\}$ can be corrupted. 
Hence, it decreases the largest votes $n_y$ at most $M$, while increasing the second-largest votes $n_{b}$ at most $M$ based on Eqn (\ref{eqn:labelcount}). 
Let $\hat{n}_y$ and $\hat{n}_b$ denote the votes of predicting the label $y$ and $b$ on the perturbed graph $\hat{G}$. We have 
$\hat{n}_y \geq n_y - M$ and $\hat{n}_b \leq n_b + M$. 
To ensure the voting classifier $\bar{f}$ still predicts $y$ for the perturbed graph $\hat{G}$, we require $\hat{n}_y > \hat{n}_b - \mathbb{I}(y<b)$ or $\hat{n}_y \geq \hat{n}_b - \mathbb{I}(y<b) + 1$, where $\mathbb{I}(y<b)$ is due to the tie breaking mechanism (i.e., we  choose a label with a smaller number when ties exist). Combining these inequalities, 
we require $n_y-M \geq n_b+M - \mathbb{I}(y<b) + 1$, yielding
\begin{align}
\label{eqn:reqclassifier}
M \leq \lfloor \frac{n_y-n_{b} + \mathbb{I}(y<b)-1}{2} \rfloor.
\end{align}

We now achieve (2): Recall $\bar{\mathcal{E}}_M $ the edges  in $\bar{\mathcal{E}}$ with top-$M$ scores in $\textbf{n}^{\gamma}$, which $\bar{\mathcal{E}}$ are edges in the complement of graph $G$. 
Similarly, with $M$ perturbed edges, the votes of every explanatory edge $e \in \mathcal{E}_k$ is decreased at most $M$, while the votes of every other edge $e \in \mathcal{E}_C \setminus \mathcal{E}_k$ is increased at most $M$ based on Eqn (\ref{eqn:votes}). 
Note that the edge $l \in \mathcal{E}_k$ has the $\lambda$-th highest votes $n_{l}^{\gamma}$, and the edge $h_M \in \bar{\mathcal{E}}_M \cup (\mathcal{E} \setminus \mathcal{E}_k)$ has the $(k-\lambda+1)$-th highest votes $n_{h_M}^{\gamma}$. 
Let $\hat{n}_{l}^{\gamma}$ and $\hat{n}_{h_M}^{\gamma}$ denote the votes of the edge $l$ and $h_M$ on $\hat{G}$ for each $M$. 
Likewise, we have $\hat{n}_{l}^{\gamma} \geq n_{l}^{\gamma} - M$ and $\hat{n}_{h_M}^{\gamma}  \leq {n}_{h_M}^{\gamma} + M$ for every $h_M$ (note $h_M$ depends on $M$). 
If the smallest votes $\hat{n}_{l}^{\gamma}$ of edge $l$ after the perturbation is still larger than the largest votes $\hat{n}_{h_M}^{\gamma}$ of the edge $h_M$, then at least $\lambda$ edges in $\mathcal{E}_k$ are still in $\hat{\mathcal{E}}_k$. 
This requires: $\hat{n}_{l}^{\gamma} > \hat{n}_{h_M}^{\gamma} - \mathbb{I}(l<h_M))$ for all $h_M$, where $\mathbb{I}(l<h_M)$ is due to the tie breaking. 
Combining these inequalities together, we require ${n}_{l}^\gamma -M > {n}_{h_M}^\gamma +M - \mathbb{I}(l<h_M), \forall M$, yielding
\begin{align}
\label{eqn:reqxGNN}
{n}_{l}^\gamma- {n}_{h_M}^\gamma + \mathbb{I}(l<h_M) > 2M
\end{align}

By satisfying both requirements, we force 
$$M \leq \min \big( \lfloor \frac{n_y-n_{b} + \mathbb{I}(y<b)-1}{2} \rfloor, M_h\big),$$
where $M_h = \max \, M,  \quad {s.t.} \quad {n}_{l}^\gamma- {n}_{h_M}^\gamma + \mathbb{I}(l<h_M) > 2M.$

\begin{algorithm}[!t] 
\small
\caption{XGNNCert: Classification, Explanation, and Certified Perturbation Size}
\label{alg:algorithm1}
\textbf{Input}: Graph $G={(\mathcal{V},\mathcal{E})}$ with $k$  explanation edges,  base classifier $f$, base explainer $g$, 
hyperparameters: ratio $p$, ratio $\gamma$, number of subgraphs $T$, hash function $h$.

\textbf{Output}: Prediction $y$, explanation $\mathcal{E}_{k}$, certified perturbation size $\{{  M_\lambda},\lambda\in[1,k]\}$ for $G$

\begin{algorithmic}[1] 
\STATE Initialize $T$ subgraphs with empty edges $\{G^{i}=(\mathcal{V},\mathcal{E}^{i}=\emptyset),i=1,\cdots,T\}$.
\STATE Initialize $T$ complete subgraphs with empty edges $\{G^{i}_{C}=(\mathcal{V},\mathcal{E}^{i}_{C}=\emptyset),i=1,\cdots,T\}$.
\STATE Initialize $T$ hybrid subgraphs with empty edges $\{G^{i}_{H}=(\mathcal{V},\mathcal{E}^{i}_{H}=\emptyset),i=1,\cdots,T\}$.
\STATE Initialize a complete edge set $\mathcal{E}_{C} = \{(u,v),\forall u,v\in\mathcal{V}:u<v\}$
\STATE Initialize votes for all classes ${\bf n}= \{0\}^{|\mathcal{C}|}$, and all edges: ${\bf n}^{\gamma}= \{0\}^{|\mathcal{E}_{C}|}$

\FOR{Edge $e \in \mathcal{E}_{C}$}  
\STATE Assign index $i_e = h[\mathrm{str}(u) + \mathrm{str}(v)] \, \, \mathrm{mod} \, \, T+1.$
\IF{$e \in G$}
\STATE Add $e$ into subgraph 
$G^{i_{e}}$ by $\mathcal{E}^{i_{e}}\cup= \{e\}$
\ENDIF
\STATE Add $e$ into complete subgraph  $G^{i_{e}}_{C}$ by $\mathcal{E}^{i_{e}}_{C}\cup=\{e\}$
\ENDFOR

\FOR{$i \in [1,T]$}
\STATE Add the $i$-th subgraph $G^i$ into $i$-th hybrid subgraph by $\mathcal{E}_{H}^{i}\cup = \mathcal{E}^{i}$
\FOR{$j\in [1,i-1]\cup[i+1,T]$}
\STATE Randomlize a value $\Tilde{p}\in[0,1)$
\IF{$\Tilde{p}\leq p$}
\STATE Add the $j$-th complete subgraph into $i$-th hybrid subgraph by $\mathcal{E}_{H}^{i}\cup = \mathcal{E}^{i}_{C}$
\ENDIF
\ENDFOR
\ENDFOR
\FOR{$G^{i}_{C},i\in [1,T]$}
\STATE Predict $G^{i}_{C}$'s label via the base classifier: $c=f(G^{i}_{C})$
\STATE Add to the classification vote by 1: $n_c+=1$
\ENDFOR
\STATE Find the class with the most votes: $y = \underset{c \in \mathcal{C}}{{\arg\max}} \, n_c$
\STATE Find the class with the second most votes: $b = \underset{c \in \mathcal{C}\setminus\{y\}}{{\arg\max}} \, n_c$
\STATE Calculate the certified bound w.r.t. the classifier: $M_{f} = \lfloor \frac{n_y-n_{b} + \mathbb{I}(y<b)-1}{2} \rfloor$.
\FOR{$G^{i}_{C},i\in [1,T]$}
\STATE Explain $G^{i}_{C}$'s output via the base explainer: $\textbf{m}^{i}=g(G_H^{i},y)$
\FOR{$e\in G^i_{H}$}
\IF{${\bf m}_{e}^{i} \geq {\bf m}_{(\gamma)}^{i}$}
\STATE $n_{e}^{\gamma}+=1$
\ENDIF
\ENDFOR
\ENDFOR
\STATE Find the edges with top-k votes in $G$: $\mathcal{E}_{k}=\mathcal{E}.top_{k}(\textbf{n}^{\gamma})$
\STATE Initialize $M=0$, $\{{  M_\lambda}=0, \lambda=1,\cdots,k\}$ 
\WHILE{$M_{1}=M$}
\STATE $M+=1$
\STATE Find the edges with top-$M$ votes in $\mathcal{E}_{C}\setminus\mathcal{E}$: $\overline{\mathcal{E}}_{M}=(\mathcal{E}_{C}\setminus\mathcal{E}).top_{M}(\textbf{n}^{\gamma})$
\FOR{$\lambda \in [1,k]$}
\STATE Find the edge $l \in \mathcal{E}_k$ is with the $\lambda$-th highest votes $n_{l}^{\gamma}$, 
 \STATE Find the edge 
  $h \in \bar{\mathcal{E}}_{M} \cup (\mathcal{E} \setminus \mathcal{E}_k)$ 
 with the $(k-\lambda+1)$-th highest votes ${n}_{h}^{\gamma}$ in ${\textbf{n}}^{\gamma}$ 
\IF{${n}_{l}^\gamma- {n}_{h}^\gamma + \mathbb{I}(l<h) > 2M$}
\STATE ${  M_\lambda}=M$
\ENDIF
\ENDFOR
\ENDWHILE
\FOR{$\lambda\in [1,k]$}
\STATE ${  M_\lambda}=\text{min}({  M_\lambda},M_{f})$
\ENDFOR
\STATE {\bf Return $y, \mathcal{E}_K, \{{  M_\lambda}, \lambda \in [1,k]\}$} 
\end{algorithmic}
\end{algorithm}

\section{Pseudo Code on XGNNCert}

{Here we provide the pseudo code of our XGNNCert, shown in Algorithm~\ref{alg:algorithm1}}.

\section{Experimental Setup and More Results}
\label{app:evaluation}

{
\subsection{Detailed Experimental Setup}
\label{app:setup}

{\bf Dataset statistics:} Table~\ref{tab:Datasets} shows the statistics of the used datasets.

{\bf Hyperparameter and network architecture details in training GNN classifiers and explainers:}  We have tested the base GNN classifiers with 2 and 3 layers, the hidden neurons $\{32, 64, 128, 192\}$, learning rates $\{0.001, 0.002, 0.005, 0.01\}$ and training epochs $\{600, 800, 1000, 1200\}$ with the Adam optimizer (no weight decay in the training). Finally, our base GNN classifiers are all 3-layer architectures with 128 hidden neurons, the learning rate as 0.001, and the epochs as 1000. 

For base GNN explainers, we simply use the configured hyperparameters in their source code. We set their hidden sizes as 64, coefficient sizes as 0.0001, coefficient entropy weights as 0.001, learning rates as 0.01, and epochs as 20. For PG Explainer, we set its first temperature as 5 and last temperature as 2. For GSAT, we set its final rate as 0.7, decay interval as 10 and decay rate as 0.1. For Refine, we set its gamma parameter as 1, beta parameter as 1 and tau parameter as 0.1.

}

 {{\bf Training the GNN classifier and GNN explainer}: Traditionally, we only use the training graphs (with their labels) to train a GNN classifier, which is used to predict the testing graphs. In {\name}, however, the voting classifier uses the hybrid subgraphs (the combination of subgraphs from the testing graphs and from the corresponding complete graphs) 
for evaluation. To enhance the testing performance of our voting classifier, we propose to train  the GNN classifier using not only the original training graphs but also the hybrid subgraphs, whose labels are same as the training graphs'\footnote{We observe the wrong prediction rate on our test hybrid subgraphs is relatively high, if we use the GNN classifier trained only on the raw training graphs. For instance, when $T=30$, the wrong prediction rate could be range from 35\% to 65\% on the five datasets. This is because the training graphs used to train the GNN classifier and test hybrid subgraphs have drastically different distributions. 
}.

Instead, the GNN  explainer is directly trained on raw clean graphs due to two reasons. First, the cost of training the explainer on subgraphs is high; Second, some subgraphs do not contain groundtruth explanatory edges, making it unable to explain these subgraphs during training.    
}

\begin{table}[!t]
    \centering
    \scriptsize
    \begin{tabular}{ccccccccc}
		\toprule
         Dataset& Graphs&$|\mathcal{V}|_{\text{avg}}$&$|\mathcal{E}|_{\text{avg}}$&Features&GT Graphs&GT Explanation&$|\mathcal{E}_{\text{GT}}|_{\text{avg}}$&k\\
         \midrule
         
         SG+House&1,000&13.69&15.56& 10&693&House Shape&6&6\\
         SG+Diamond& 1,000&10.46&14.03&10&486&Diamond Shape&5&5\\
         SG+Wheel&1,000& 12.76&14.07 & 10&589&Wheel Shape&8&8\\
         \midrule
         Benzene& 12,000&20.58&43.65&14&6,001 &Benzene Ring&6&6\\
         Fluoride-Carbonyl (FC) & 8,6716&21.36&45.37&14& 3,254 &F- and C=O& 5 & 5 \\
         \bottomrule
    \end{tabular}
    \vspace{-2mm}
    \caption{Datasets and their statistics.}
    \vspace{-4mm}
    \label{tab:Datasets}
\end{table}

{\subsection{More Results}}
\label{app:moreresults}

{ 
Figure~\ref{fig:CEA_T_pge_2}---Figure~\ref{fig:CEA_T_gsat_2} show the certified perturbation size vs. $\lambda$ on the three GNN explainers. 

Figure~\ref{fig:CEA_p_PG-2}---Figure~\ref{fig:CEA_h_PG} show the certified perturbation size of {\name} vs. $p, \gamma, h$ on PGExplainer, respectively. Additionally, Table \ref{tab:frac_performance}  and Table \ref{tab:hash_performance}  show the prediction accuracy of {\name}  vs. $p$ and $h$, respectively. We see the results are close, implying {\name}  is insensitive to $p$ and $h$. 
}

Figure~\ref{fig:deceive} shows the explanation results when the GNN model is deceived. We see that explaining wrong predictions yields explanation results that are not meaningful.

\begin{table}[!t]
\centering
\begin{tabular}{lccccc}
\toprule
\textbf{Ratio $p$}   & \textbf{SG+House} & \textbf{SG+Diamond} & \textbf{SG+Wheel} & \textbf{Benzene} & \textbf{FC} \\ \midrule
0.0              & 0.900             & 0.925               & 0.905             & 0.723           & 0.674       \\
0.02             & 0.895             & 0.920               & 0.900             & 0.723           & 0.692       \\
0.03             & 0.905             & 0.935               & 0.900             & 0.723           & 0.692       \\
0.04             & 0.895             & 0.925               & 0.900             & 0.725           & 0.662       \\ \bottomrule
\end{tabular}
\caption{Prediction accuracy of XGNNCert 
 with different $\rho$ (default $p=0.3$).}
\label{tab:frac_performance}
\end{table}

\begin{table}[!t]
\centering
\begin{tabular}{lccccc}
\toprule
\textbf{Hash function $h$}   & \textbf{SG+House} & \textbf{SG+Diamond} & \textbf{SG+Wheel} & \textbf{Benzene} & \textbf{FC} \\ \midrule
\textbf{MD5}    & 0.905             & 0.935               & 0.900             & 0.723           & 0.692       \\

\textbf{SHA1}     & 0.905             & 0.935               & 0.895             & 0.718           & 0.692       \\
\textbf{SHA256}  & 0.900             & 0.935               & 0.905             & 0.725           & 0.674       \\ \bottomrule
\end{tabular}
\caption{Prediction accuracy of XGNNCert 
 with different hash functions. Default is "MDS".}
 \vspace{-5mm}
\label{tab:hash_performance}
\end{table}

\begin{figure*}[!t]
	\centering
 \subfloat[{SG+Wheel}]
	{\centering\includegraphics[scale=0.22]{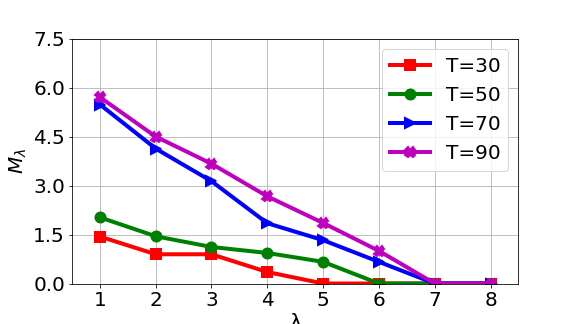}}
 \,
 \subfloat[{FC}]
	{\centering \includegraphics[scale=0.22]{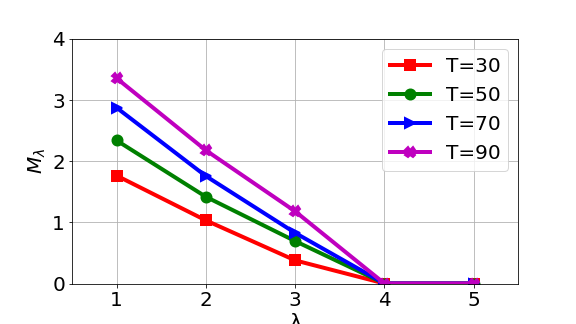}} 
	\caption{Certified perturbation size over all testing graphs  vs. $\lambda$ on PGExplainer.} 
	\label{fig:CEA_T_pge_2}
	 \vspace{-2mm}
\end{figure*}

\begin{figure*}[!t]
	\centering
 \subfloat[{SG+House}]
	{\centering\includegraphics[scale=0.22]{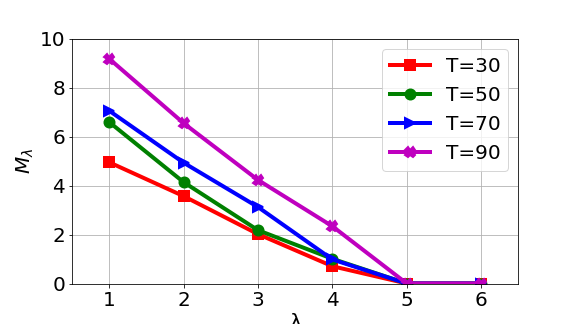}}
 \,
	\subfloat[{SG+Diamond}]
	{\centering \includegraphics[scale=0.22]{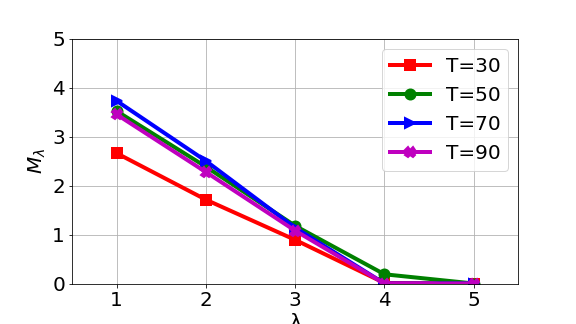}}
\,
 \subfloat[{SG+Wheel}]
	{\centering\includegraphics[scale=0.22]{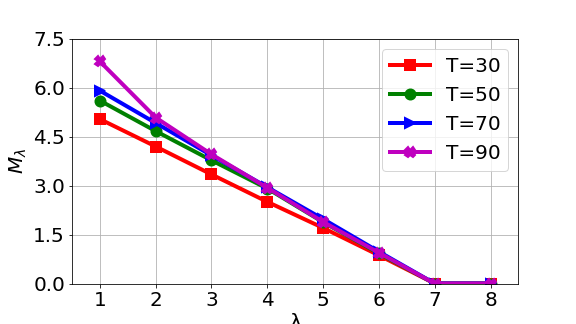}}
 \,
 \vspace{-4mm}
\\
 \subfloat[{Beneze}]
	{\centering\includegraphics[scale=0.22]{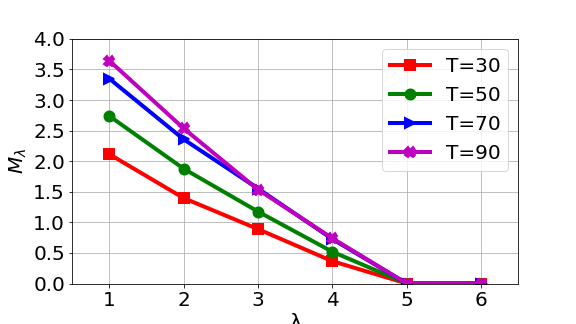}}
\,	
 \subfloat[{FC}]
	{\centering \includegraphics[scale=0.22]{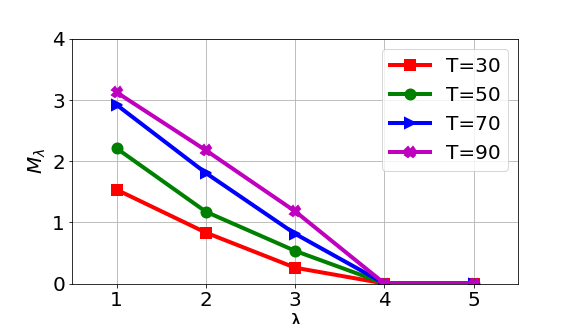}} 
	\caption{Certified perturbation size over all testing graphs vs. $\lambda$ on ReFine.} 
	\label{fig:CEA_T_refine-2}
	 \vspace{-2mm}
\end{figure*}

  \begin{figure*}[!t]
	\centering
 \subfloat[{SG+House}]
	{\centering\includegraphics[scale=0.22]{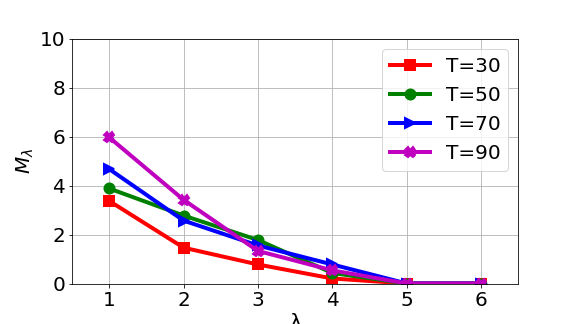}}
 \,
	\subfloat[{SG+Diamond}]
	{\centering \includegraphics[scale=0.22]{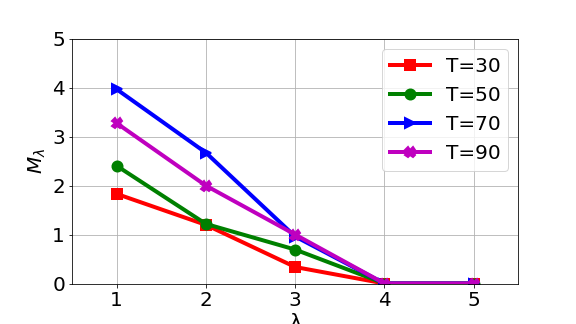}}
\,
 \subfloat[{SG+Wheel}]
	{\centering\includegraphics[scale=0.22]{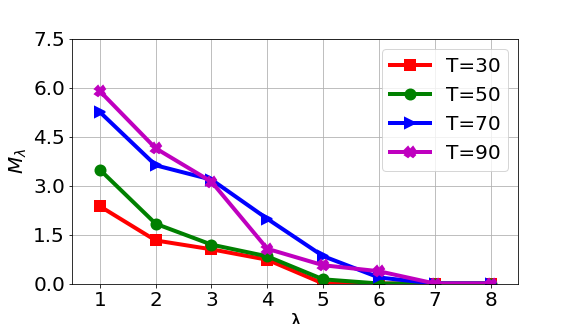}}
 \,
  \vspace{-4mm}
\\
 \subfloat[{Beneze}]
	{\centering\includegraphics[scale=0.22]{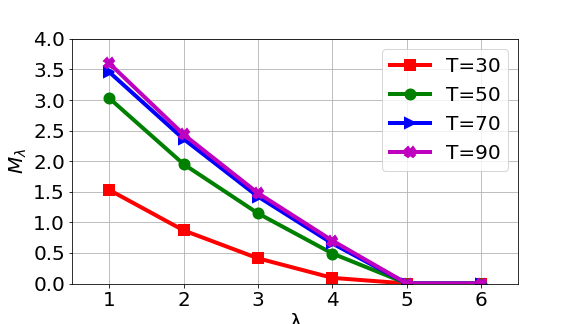}}
\,	
 \subfloat[{FC}]
	{\centering \includegraphics[scale=0.22]{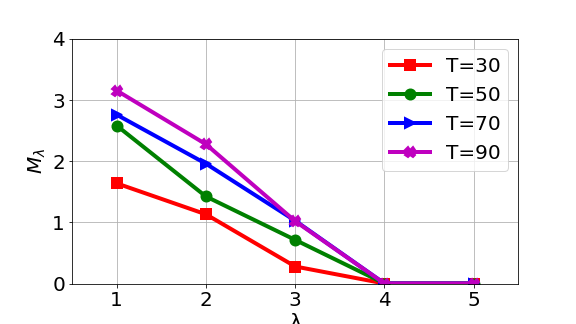}} 
	\caption{{  Maximal}
 certified perturbation size over all testing graphs vs. $\lambda$ on GSAT.} 
	\label{fig:CEA_T_gsat_2}
	 \vspace{-2mm}
\end{figure*}

\begin{figure*}[!t]
	\centering
 \subfloat[{SG+Wheel}]
	{\centering\includegraphics[scale=0.22]{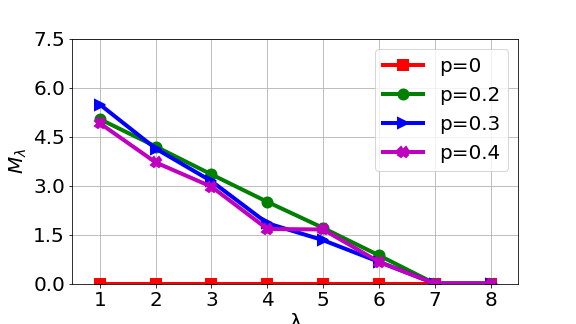}}
 \,
\subfloat[{FC}]
	{\centering \includegraphics[scale=0.22]{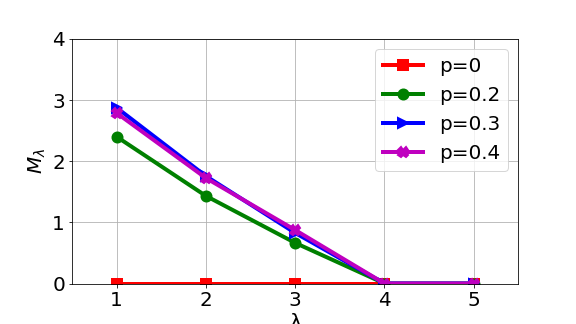}} 
	\caption{Certified perturbation size over all testing graphs  vs. $p$ on PGExplainer.} 
	\label{fig:CEA_p_PG-2}
	 \vspace{-6mm}
\end{figure*}

\begin{figure*}[!t]
	\centering
 \subfloat[{SG+Wheel}]
	{\centering\includegraphics[scale=0.22]{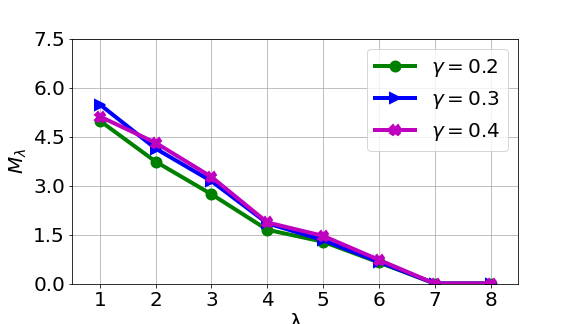}}
\,	
 \subfloat[{FC}]
	{\centering \includegraphics[scale=0.22]{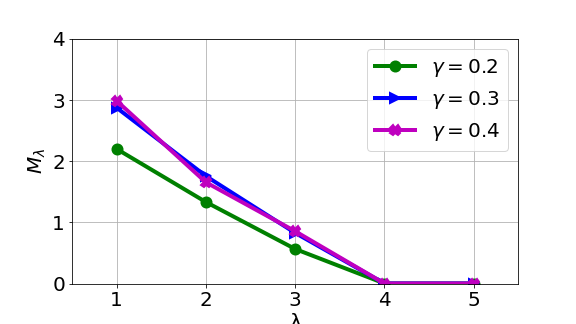}} 
	\caption{Certified perturbation size over all testing graphs vs. $\gamma$ on PGExplainer.} 
	\label{fig:CEA_gamma_PG-2}
	 \vspace{-6mm}
\end{figure*}

\begin{figure*}[!t]
	\centering
 \subfloat[{SG+House}]
	{\centering\includegraphics[scale=0.22]{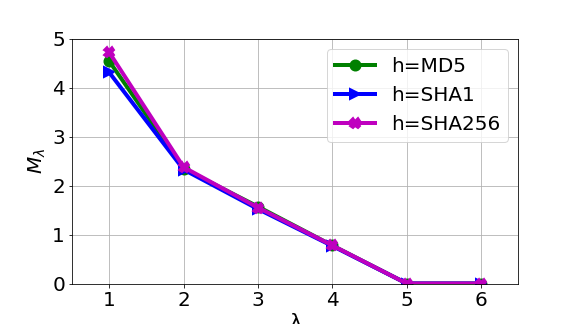}}
 \,
	\subfloat[{SG+Diamond}]
	{\centering \includegraphics[scale=0.22]{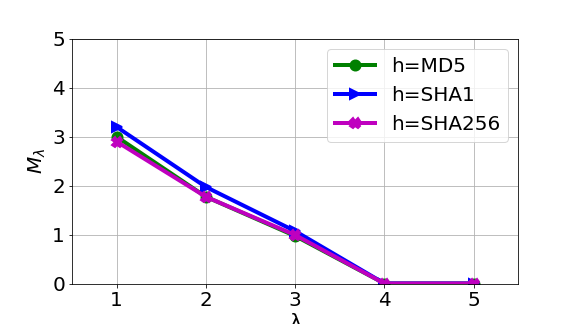}}
\,
 \subfloat[{SG+Wheel}]
	{\centering\includegraphics[scale=0.22]{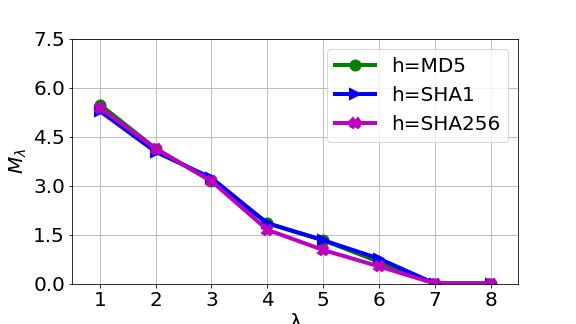}}
 \,
 \vspace{-4mm}
\\
 \subfloat[{Beneze}]
	{\centering\includegraphics[scale=0.22]{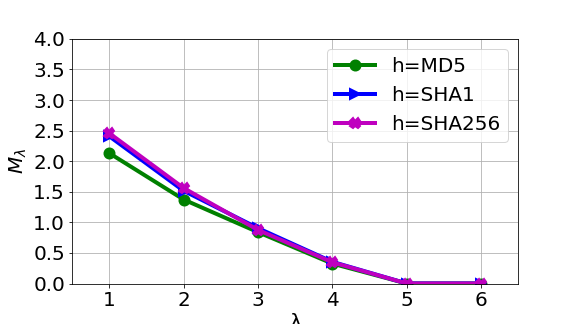}}
\,	
 \subfloat[{FC}]
	{\centering \includegraphics[scale=0.22]{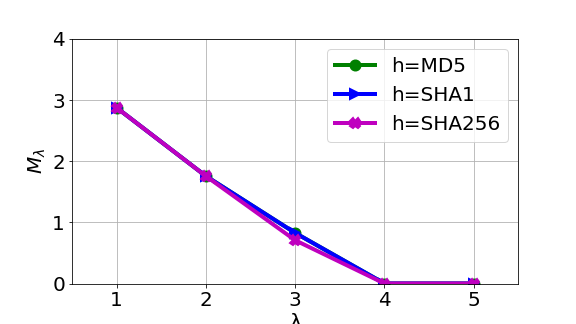}} 
	\caption{Certified perturbation size over all testing graphs vs. $h$ on PGExplainer.} 
	\label{fig:CEA_h_PG}
	 \vspace{-2mm}
\end{figure*}

\begin{figure*}[t]
	\centering
    \captionsetup[subfloat]{labelsep=none,format=plain,labelformat=empty,farskip=0pt}
\subfloat[{Ground-Truth}]{
 \subfloat[{House}]{
	{\centering\includegraphics[scale=0.16]{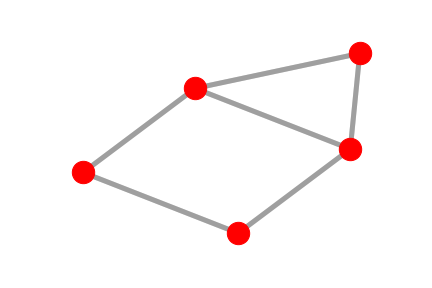}}}
 \,
 \subfloat[{Diamond}]{
	{\centering\includegraphics[scale=0.16]{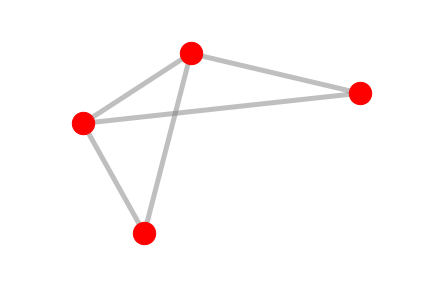}}}
 \,
 \subfloat[{Wheel}]{
	{\centering\includegraphics[scale=0.16]{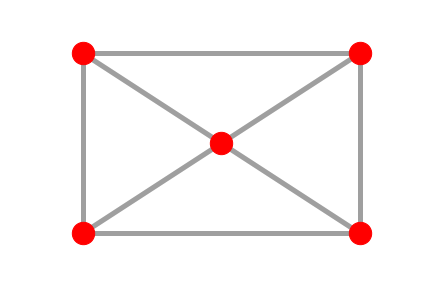}}}
 
 \,
 \subfloat[{Benzene}]{
	{\centering\includegraphics[scale=0.16]{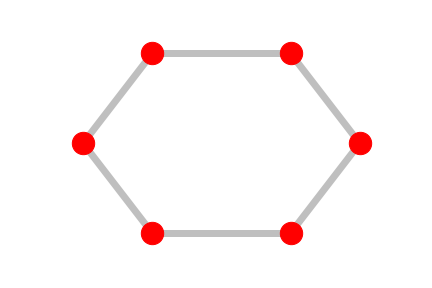}}}
 
 \,
 \subfloat[{FC}]{
	{\centering\includegraphics[scale=0.16]{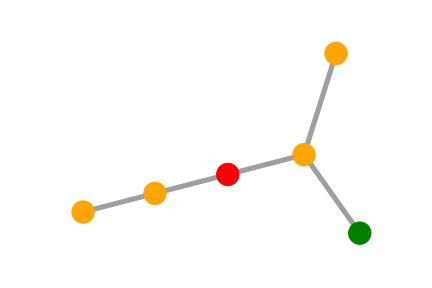}}}
 }
 \\
\subfloat[{Original Explanation}]{
 \subfloat[{House}]{
	{\centering\includegraphics[scale=0.16]{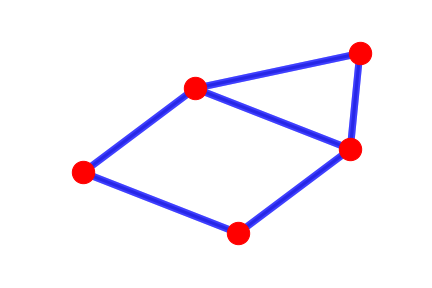}}}
 \,
 \subfloat[{Diamond}]{
	{\centering\includegraphics[scale=0.16]{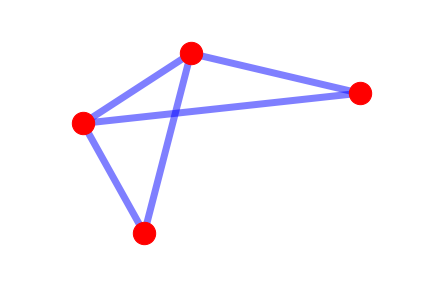}}}
 \,
 \subfloat[{Wheel}]{
	{\centering\includegraphics[scale=0.16]{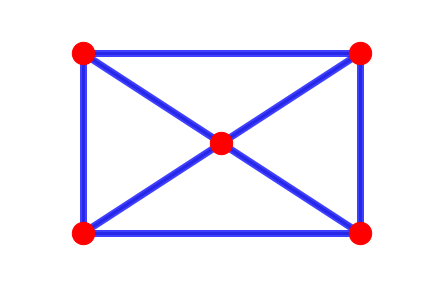}}}
 
 \,
 \subfloat[{Benzene}]{
	{\centering\includegraphics[scale=0.16]{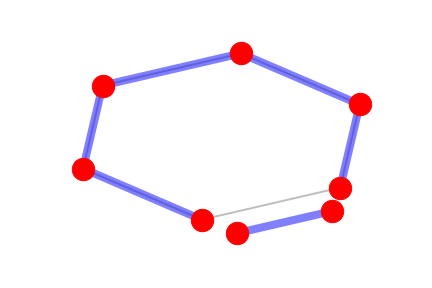}}}
 
 \,
 \subfloat[{FC}]{
	{\centering\includegraphics[scale=0.16]{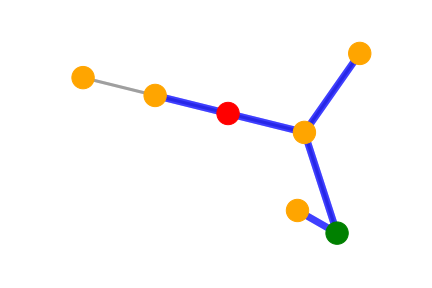}}}
 }
 \\
\subfloat[{GNN Deceived Explanation}]{
 \subfloat[{House}]{
	{\centering\includegraphics[scale=0.16]{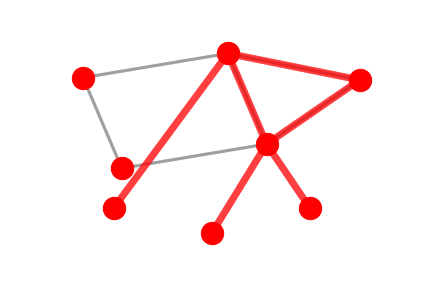}}}
 \,
 \subfloat[{Diamond}]{
	{\centering\includegraphics[scale=0.16]{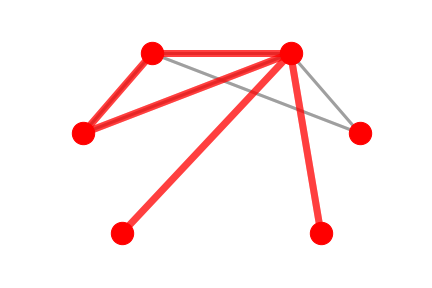}}}
 \,
 \subfloat[{Wheel}]{
	{\centering\includegraphics[scale=0.16]{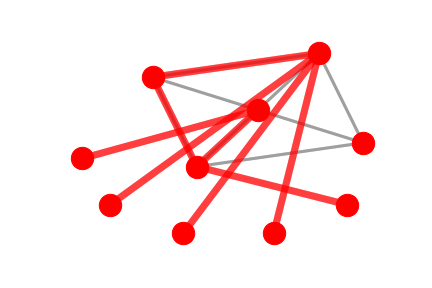}}}
 
 \,
 \subfloat[{Benzene}]{
	{\centering\includegraphics[scale=0.16]{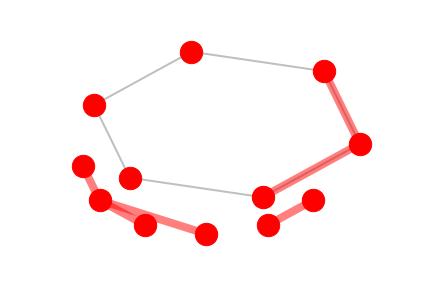}}}
 
 \,
 \subfloat[{FC}]{
	{\centering\includegraphics[scale=0.16]{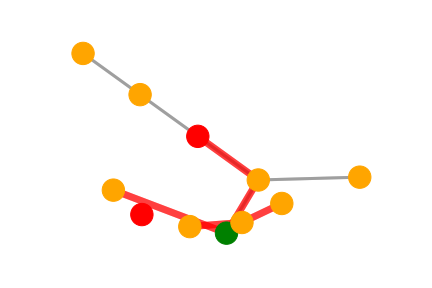}}}
 }
	\caption{Examples of how an explanatory subgraph outputted by PGExplainer changes when GNN is deceived. Top Row: Groundtruth Explanation; Middle Row: Explanation under correct predictions; Bottom Row: Explanation when GNN is deceived by a graph perturbation (2 edges are perturbed). 
 } 
	\vspace{-4mm}
  \label{fig:deceive}
\end{figure*}

\section{Discussion}
\label{app:discussion}

{\bf Instability of GNN explainers:}
We conduct experiments on the well-known GNNExplainer~\citep{GNNEx19} to show its unstable explanation results. Particularly, we run it 5 times and show the explanation results in Table~\ref{tab:instable}, where “Std” is the Standard Deviation of the explanation accuracy on test data across the 5 runs, and “Change Rate” is the average fraction of different explanation edges among every pair of the 5 runs. 
We can see both the variance and change rate are  large, meaning it is unreliable and difficult to pick any run of the result to design the robust explainer. 

\begin{table}[!h]
    \centering
    \begin{tabular}{c|c|c|c|c|c}
    \toprule
    Dataset&SG+House&SG+Diamond&SG+Wheel&Benzene&FC\\
    
    \Xhline{0.8pt} 
    Exp. Accuracy&0.624&0.368&0.475&0.276&0.226 \\
    \cline{1-6}
    Std& 9.3\%&9.7\%&8.1\%&10.9\%&12.4\%\\
    \cline{1-6}
    Change Rate&36.8\%&64.0\%&51.6\%&72.6\%&76.9\%\\
    \bottomrule
    \end{tabular}
    \caption{Instability of GNNExplainer}
    \label{tab:instable}
\end{table}

\begin{table}[!ht]
    \centering
    \renewcommand{\arraystretch}{1.2}
    \begin{tabular}{c|c|c|c|c}
        \hline
        \textbf{Dataset} & \textbf{Pred. Acc. (Avg)} & \textbf{Pred. Acc. (Std)} & \textbf{Exp. Acc. (Avg)} & \textbf{Exp. Acc. (Std)} \\ \hline
        Benzene & 0.722 & 0.002 & 0.466 & 0.007 \\ \hline
        FC      & 0.682 & 0.012 & 0.358 & 0.037 \\ \hline
    \end{tabular}
    \caption{Averaged prediction and explanation accuracy of XGNNCert on the two real-world datasets with 5 random node orderings.}
    \label{tab:gnncert_runs}
\end{table}

{\bf Node-order invariant vs. variant GNNs:} There exist both node-order invariant GNNs  (whose outputs are insensitive to the node ordering) and node-order variant GNNs  (whose outputs depend on the node ordering). Node-order invariant GNNs typically use, e.g., the mean and convolution aggregator such as  GCN~\citep{kipf2017semi}, SGC~\citep{wu2019simplifying}, GIN~\citep{xu2018powerful}, GAT~\citep{velivckovic2018graph}, GSAGE-mean~\citep{hamilton2017inductive}). 
Node-order variant GNNs are based on, e.g., random neighbor sampling~\citep{papp2021dropgnn,Rong2020DropEdge,Zeng2020GraphSAINT}, LSTM aggregator~\citep{hamilton2017inductive}, relational pooling~\citep{murphy2019relational}, positional embedding ~\citep{dwivedi2022graph,kreuzer2021rethinking,zhu2023hierarchical}. 

While node-order invariant GNNs are desirable in certain cases, recent works ~\citep{Loukas2020What,papp2021dropgnn,huang2022going} show node-order variant GNNs can produce better expressivity. This ranges from the classic GSAGE with LSTM to modern graph transformers~\citep{kreuzer2021rethinking,zhu2023hierarchical}. Our GNN voting classifier is node-order variant due to the property of hash function.  

{To further explore the impact of node permutations on {\name}, we randomly permute the input graphs 5 times and report {\name}'s average prediction and explanation accuracies on the two real-world datasets under the default setting in Table~\ref{tab:gnncert_runs}.  
We observe that {\name} exhibits stable prediction and explanation accuracies across the 5 runs. This demonstrates that, though {\name} is not inherently permutation invariant, its classification and explanation performance remain relatively stable to node permutations. 
We hypothesize that one possible reason for this stability is that {\name} augments the training graphs with a set of subgraphs to train the GNN classifier. This augmentation may mitigate the effect of node ordering, as the subgraphs are much smaller in size.}

{\bf Can this framework be extended to node-level or edge-level tasks?} Theoretically, it is possible, but needs technique adaptation. For example, in the node-level task, we are given a target node and its prediction by a GNN model, then GNN explainers aim to find the subgraph (usually from the target node’s neighboring graph) that is most important for the target node’s prediction. When applying the proposed framework for certifying node-level explainers, it becomes designing a graph division and voting strategy such that: with an arbitrary graph perturbation under a perturbation budget, 1) the voting classifier guarantees the correct prediction for \emph{the target node} on the perturbed graph, and 2) the voting explainer guarantees the explanation results on the perturbed graph and clean graph are close. The current graph division strategy is not applicable as all subgraphs have disjoint nodes, while the target node should be contained in all subgraphs for the node-level task. Hence, a key challenge is how to adapt the graph division and voting strategy to satisfy 1) and 2), particularly guaranteeing only a bounded number of subgraphs is affected when predicting the target node, while the explanations of these subgraphs' predictions are also retained. We acknowledge it is interesting future work to extend the proposed framework specially for node/edge-level explanation tasks.

\clearpage
\newpage

%% file: iclr2025_conference.bbl
\begin{thebibliography}{80}
\providecommand{\natexlab}[1]{#1}
\providecommand{\url}[1]{\texttt{#1}}
\expandafter\ifx\csname urlstyle\endcsname\relax
  \providecommand{\doi}[1]{doi: #1}\else
  \providecommand{\doi}{doi: \begingroup \urlstyle{rm}\Url}\fi

\bibitem[Agarwal et~al.(2023)Agarwal, Queen, Lakkaraju, and Zitnik]{agarwal2023evaluating}
Chirag Agarwal, Owen Queen, Himabindu Lakkaraju, and Marinka Zitnik.
\newblock Evaluating explainability for graph neural networks.
\newblock \emph{Scientific Data}, 10\penalty0 (1):\penalty0 144, 2023.

\bibitem[Bajaj et~al.(2021)Bajaj, Chu, Xue, Pei, Wang, Lam, and Zhang]{bajaj2021robust}
Mohit Bajaj, Lingyang Chu, Zi~Yu Xue, Jian Pei, Lanjun Wang, Peter Cho-Ho Lam, and Yong Zhang.
\newblock Robust counterfactual explanations on graph neural networks.
\newblock 34:\penalty0 5644--5655, 2021.

\bibitem[Baldassarre \& Azizpour(2019)Baldassarre and Azizpour]{baldassarre2019explainability}
Federico Baldassarre and Hossein Azizpour.
\newblock Explainability techniques for graph convolutional networks.
\newblock \emph{ICML Workshop}, 2019.

\bibitem[Behnam \& Wang(2024)Behnam and Wang]{behnam2024graph}
Arman Behnam and Binghui Wang.
\newblock Graph neural network causal explanation via neural causal models.
\newblock In \emph{ECCV}, 2024.

\bibitem[Bojchevski et~al.(2020)Bojchevski, Gasteiger, and G{\"u}nnemann]{bojchevski2020efficient}
Aleksandar Bojchevski, Johannes Gasteiger, and Stephan G{\"u}nnemann.
\newblock Efficient robustness certificates for discrete data: Sparsity-aware randomized smoothing for graphs, images and more.
\newblock In \emph{ICML}, 2020.

\bibitem[Carlini et~al.(2019)Carlini, Athalye, Papernot, Brendel, Rauber, Tsipras, Goodfellow, Madry, and Kurakin]{carlini2019evaluating}
Nicholas Carlini, Anish Athalye, Nicolas Papernot, Wieland Brendel, Jonas Rauber, Dimitris Tsipras, Ian Goodfellow, Aleksander Madry, and Alexey Kurakin.
\newblock On evaluating adversarial robustness.
\newblock \emph{arXiv}, 2019.

\bibitem[Cohen et~al.(2019)Cohen, Rosenfeld, and Kolter]{cohen2019certified}
Jeremy~M Cohen, Elan Rosenfeld, and J~Zico Kolter.
\newblock Certified adversarial robustness via randomized smoothing.
\newblock \emph{arXiv preprint arXiv:1902.02918}, 2019.

\bibitem[Dai et~al.(2018)Dai, Li, Tian, Huang, Wang, Zhu, and Song]{dai2018adversarial}
Hanjun Dai, Hui Li, Tian Tian, Xin Huang, Lin Wang, Jun Zhu, and Le~Song.
\newblock Adversarial attack on graph structured data.
\newblock In \emph{ICML}, 2018.

\bibitem[Dwivedi et~al.(2022)Dwivedi, Luu, Laurent, Bengio, and Bresson]{dwivedi2022graph}
Vijay~Prakash Dwivedi, Anh~Tuan Luu, Thomas Laurent, Yoshua Bengio, and Xavier Bresson.
\newblock Graph neural networks with learnable structural and positional representations.
\newblock In \emph{ICLR}, 2022.

\bibitem[Feng et~al.(2022)Feng, Liu, Yang, Tang, Du, and Hu]{feng2023degree}
Qizhang Feng, Ninghao Liu, Fan Yang, Ruixiang Tang, Mengnan Du, and Xia Hu.
\newblock Degree: Decomposition based explanation for graph neural networks.
\newblock In \emph{ICLR}, 2022.

\bibitem[Funke et~al.(2022)Funke, Khosla, Rathee, and Anand]{funke2022zorro}
Thorben Funke, Megha Khosla, Mandeep Rathee, and Avishek Anand.
\newblock Zorro: Valid, sparse, and stable explanations in graph neural networks.
\newblock \emph{IEEE TKDE}, 2022.

\bibitem[Hamilton et~al.(2017)Hamilton, Ying, and Leskovec]{hamilton2017inductive}
Will Hamilton, Zhitao Ying, and Jure Leskovec.
\newblock Inductive representation learning on large graphs.
\newblock In \emph{NIPS}, 2017.

\bibitem[Hammoudeh \& Lowd(2023)Hammoudeh and Lowd]{hammoudeh2023feature}
Zayd Hammoudeh and Daniel Lowd.
\newblock Feature partition aggregation: A fast certified defense against a union of l\_0 attacks.
\newblock In \emph{The Second Workshop on New Frontiers in Adversarial ML}, 2023.

\bibitem[Hong et~al.(2022)Hong, Wang, and Hong]{hong2022unicr}
Hanbin Hong, Binghui Wang, and Yuan Hong.
\newblock Unicr: Universally approximated certified robustness via randomized smoothing.
\newblock In \emph{ECCV}, 2022.

\bibitem[Huang et~al.(2022)Huang, Wang, Li, and He]{huang2022going}
Zhongyu Huang, Yingheng Wang, Chaozhuo Li, and Huiguang He.
\newblock Going deeper into permutation-sensitive graph neural networks.
\newblock In \emph{ICML}, pp.\  9377--9409. PMLR, 2022.

\bibitem[Jia et~al.(2020)Jia, Wang, Cao, and Gong]{jia2020certified}
Jinyuan Jia, Binghui Wang, Xiaoyu Cao, and Neil~Zhenqiang Gong.
\newblock Certified robustness of community detection against adversarial structural perturbation via randomized smoothing.
\newblock In \emph{Proceedings of The Web Conference 2020}, 2020.

\bibitem[Jia et~al.(2021)Jia, Cao, and Gong]{jia2021intrinsic}
Jinyuan Jia, Xiaoyu Cao, and Neil~Zhenqiang Gong.
\newblock Intrinsic certified robustness of bagging against data poisoning attacks.
\newblock In \emph{Proceedings of the AAAI Conference on Artificial Intelligence}, 2021.

\bibitem[Jia et~al.(2022)Jia, Liu, Cao, and Gong]{jia2022certified}
Jinyuan Jia, Yupei Liu, Xiaoyu Cao, and Neil~Zhenqiang Gong.
\newblock Certified robustness of nearest neighbors against data poisoning and backdoor attacks.
\newblock In \emph{Proceedings of the AAAI Conference on Artificial Intelligence}, 2022.

\bibitem[Jin et~al.(2020)Jin, Shi, Peruri, and Zhang]{jin2020certified}
Hongwei Jin, Zhan Shi, Venkata Jaya Shankar~Ashish Peruri, and Xinhua Zhang.
\newblock Certified robustness of graph convolution networks for graph classification under topological attacks.
\newblock In \emph{NeurIPS}, 2020.

\bibitem[Kipf \& Welling(2017)Kipf and Welling]{kipf2017semi}
Thomas~N Kipf and Max Welling.
\newblock Semi-supervised classification with graph convolutional networks.
\newblock In \emph{ICLR}, 2017.

\bibitem[Kreuzer et~al.(2021)Kreuzer, Beaini, Hamilton, L{\'e}tourneau, and Tossou]{kreuzer2021rethinking}
Devin Kreuzer, Dominique Beaini, Will Hamilton, Vincent L{\'e}tourneau, and Prudencio Tossou.
\newblock Rethinking graph transformers with spectral attention.
\newblock In \emph{NeurIPS}, 2021.

\bibitem[Lecuyer et~al.(2019)Lecuyer, Atlidakis, Geambasu, Hsu, and Jana]{lecuyer2019certified}
Mathias Lecuyer, Vaggelis Atlidakis, Roxana Geambasu, Daniel Hsu, and Suman Jana.
\newblock Certified robustness to adversarial examples with differential privacy.
\newblock In \emph{2019 IEEE Symposium on Security and Privacy (SP)}, pp.\  656--672. IEEE, 2019.

\bibitem[Levine \& Feizi(2020{\natexlab{a}})Levine and Feizi]{levine2020deep}
Alexander Levine and Soheil Feizi.
\newblock Deep partition aggregation: Provable defenses against general poisoning attacks.
\newblock In \emph{ICLR}, 2020{\natexlab{a}}.

\bibitem[Levine \& Feizi(2020{\natexlab{b}})Levine and Feizi]{levine2020randomized}
Alexander Levine and Soheil Feizi.
\newblock (de) randomized smoothing for certifiable defense against patch attacks.
\newblock In \emph{NeurIPS}, 2020{\natexlab{b}}.

\bibitem[Levine \& Feizi(2020{\natexlab{c}})Levine and Feizi]{levine2020robustness}
Alexander Levine and Soheil Feizi.
\newblock Robustness certificates for sparse adversarial attacks by randomized ablation.
\newblock In \emph{AAAI}, 2020{\natexlab{c}}.

\bibitem[Levine et~al.(2019)Levine, Singla, and Feizi]{levine2019certifiably}
Alexander Levine, Sahil Singla, and Soheil Feizi.
\newblock Certifiably robust interpretation in deep learning.
\newblock \emph{arXiv preprint arXiv:1905.12105}, 2019.

\bibitem[Li \& Wang(2025)Li and Wang]{li2025agnncert}
Jiate Li and Binghui Wang.
\newblock Agnncert: Defending graph neural networks against arbitrary perturbations with deterministic certification.
\newblock In \emph{USENIX Security}, 2025.

\bibitem[Li et~al.(2024)Li, Pang, Dong, Jia, and Wang]{li2024graph}
Jiate Li, Meng Pang, Yun Dong, Jinyuan Jia, and Binghui Wang.
\newblock Graph neural network explanations are fragile.
\newblock In \emph{ICML}, 2024.

\bibitem[Lin et~al.(2021)Lin, Lan, and Li]{GEM}
Wanyu Lin, Hao Lan, and Baochun Li.
\newblock Generative causal explanations for graph neural networks.
\newblock In \emph{ICML}, 2021.

\bibitem[Liu et~al.(2022)Liu, Chen, Liu, Xia, and Gan]{liu2022certifiably}
Ao~Liu, Xiaoyu Chen, Sijia Liu, Lirong Xia, and Chuang Gan.
\newblock Certifiably robust interpretation via r{\'e}nyi differential privacy.
\newblock \emph{Artificial Intelligence}, 313:\penalty0 103787, 2022.

\bibitem[Loukas(2020)]{Loukas2020What}
Andreas Loukas.
\newblock What graph neural networks cannot learn: depth vs width.
\newblock In \emph{ICLR}, 2020.

\bibitem[Luo et~al.(2020)Luo, Cheng, Xu, Yu, Zong, Chen, and Zhang]{DBLP:journals/corr/abs-2011-04573/PGExplainer}
Dongsheng Luo, Wei Cheng, Dongkuan Xu, Wenchao Yu, Bo~Zong, Haifeng Chen, and Xiang Zhang.
\newblock Parameterized explainer for graph neural network.
\newblock In \emph{NeurIPS}, 2020.

\bibitem[Ma et~al.(2020)Ma, Ding, and Mei]{ma2020towards}
Jiaqi Ma, Shuangrui Ding, and Qiaozhu Mei.
\newblock Towards more practical adversarial attacks on graph neural networks.
\newblock In \emph{NeurIPS}, 2020.

\bibitem[Miao et~al.(2022)Miao, Liu, and Li]{DBLP:journals/corr/abs-2201-12987/GSAT}
Siqi Miao, Miaoyuan Liu, and Pan Li.
\newblock Interpretable and generalizable graph learning via stochastic attention mechanism.
\newblock In \emph{ICML}, 2022.

\bibitem[Mu et~al.(2021)Mu, Wang, Li, Sun, Xu, and Liu]{mu2021hard}
Jiaming Mu, Binghui Wang, Qi~Li, Kun Sun, Mingwei Xu, and Zhuotao Liu.
\newblock A hard label black-box adversarial attack against graph neural networks.
\newblock In \emph{CCS}, 2021.

\bibitem[Mujkanovic et~al.(2022)Mujkanovic, Geisler, G{\"u}nnemann, and Bojchevski]{mujkanovic2022defenses}
Felix Mujkanovic, Simon Geisler, Stephan G{\"u}nnemann, and Aleksandar Bojchevski.
\newblock Are defenses for graph neural networks robust?
\newblock 35:\penalty0 8954--8968, 2022.

\bibitem[Murphy et~al.(2019)Murphy, Srinivasan, Rao, and Ribeiro]{murphy2019relational}
Ryan Murphy, Balasubramaniam Srinivasan, Vinayak Rao, and Bruno Ribeiro.
\newblock Relational pooling for graph representations.
\newblock In \emph{ICML}, 2019.

\bibitem[Papp et~al.(2021)Papp, Martinkus, Faber, and Wattenhofer]{papp2021dropgnn}
P{\'a}l~Andr{\'a}s Papp, Karolis Martinkus, Lukas Faber, and Roger Wattenhofer.
\newblock Dropgnn: Random dropouts increase the expressiveness of graph neural networks.
\newblock In \emph{NeurIPS}, 2021.

\bibitem[Pereira et~al.(2023)Pereira, Nascimento, Resck, Mesquita, and Souza]{pereira2023distill}
Tamara Pereira, Erik Nascimento, Lucas~E Resck, Diego Mesquita, and Amauri Souza.
\newblock Distill n’explain: explaining graph neural networks using simple surrogates.
\newblock In \emph{AISTATS}, 2023.

\bibitem[Pfeifer et~al.(2022)Pfeifer, Saranti, and Holzinger]{pfeifer2022gnn}
Bastian Pfeifer, Anna Saranti, and Andreas Holzinger.
\newblock Gnn-subnet: disease subnetwork detection with explainable graph neural networks.
\newblock \emph{Bioinformatics}, 38, 2022.

\bibitem[Pope et~al.(2019)Pope, Kolouri, Rostami, Martin, and Hoffmann]{pope2019explainability}
Phillip~E Pope, Soheil Kolouri, Mohammad Rostami, Charles~E Martin, and Heiko Hoffmann.
\newblock Explainability methods for graph convolutional neural networks.
\newblock In \emph{CVPR}, 2019.

\bibitem[Rath et~al.(2021)Rath, Morales, and Srivastava]{rath2021scarlet}
Bhavtosh Rath, Xavier Morales, and Jaideep Srivastava.
\newblock Scarlet: explainable attention based graph neural network for fake news spreader prediction.
\newblock In \emph{PAKDD}, 2021.

\bibitem[Rong et~al.(2020)Rong, Huang, Xu, and Huang]{Rong2020DropEdge}
Yu~Rong, Wenbing Huang, Tingyang Xu, and Junzhou Huang.
\newblock Dropedge: Towards deep graph convolutional networks on node classification.
\newblock In \emph{ICLR}, 2020.

\bibitem[Schnake et~al.(2021)Schnake, Eberle, Lederer, Nakajima, Sch{\"u}tt, M{\"u}ller, and Montavon]{schnake2021higher}
Thomas Schnake, Oliver Eberle, Jonas Lederer, Shinichi Nakajima, Kristof~T Sch{\"u}tt, Klaus-Robert M{\"u}ller, and Gr{\'e}goire Montavon.
\newblock Higher-order explanations of graph neural networks via relevant walks.
\newblock \emph{IEEE TPAMI}, 2021.

\bibitem[Shan et~al.(2021)Shan, Shen, Zhang, Li, and Li]{shan2021reinforcement/RGExplainer}
Caihua Shan, Yifei Shen, Yao Zhang, Xiang Li, and Dongsheng Li.
\newblock Reinforcement learning enhanced explainer for graph neural networks.
\newblock In \emph{NeurIPS 2021}, December 2021.

\bibitem[Sui et~al.(2022)Sui, Wang, Wu, Lin, He, and Chua]{sui2022causal}
Yongduo Sui, Xiang Wang, Jiancan Wu, Min Lin, Xiangnan He, and Tat-Seng Chua.
\newblock Causal attention for interpretable and generalizable graph classification.
\newblock In \emph{KDD}, 2022.

\bibitem[Tan \& Tian(2023)Tan and Tian]{tan2023robust}
Zeren Tan and Yang Tian.
\newblock Robust explanation for free or at the cost of faithfulness.
\newblock In \emph{ICML}, 2023.

\bibitem[Veli{\v{c}}kovi{\'c} et~al.(2018)Veli{\v{c}}kovi{\'c}, Cucurull, Casanova, Romero, Lio, and Bengio]{velivckovic2018graph}
Petar Veli{\v{c}}kovi{\'c}, Guillem Cucurull, Arantxa Casanova, Adriana Romero, Pietro Lio, and Yoshua Bengio.
\newblock Graph attention networks.
\newblock In \emph{ICLR}, 2018.

\bibitem[Vu \& Thai(2020)Vu and Thai]{vu2020pgm}
Minh Vu and My~T Thai.
\newblock Pgm-explainer: Probabilistic graphical model explanations for graph neural networks.
\newblock In \emph{NeurIPS}, 2020.

\bibitem[Wang \& Gong(2019)Wang and Gong]{wang2019attacking}
Binghui Wang and Neil~Zhenqiang Gong.
\newblock Attacking graph-based classification via manipulating the graph structure.
\newblock In \emph{CCS}, 2019.

\bibitem[Wang et~al.(2021{\natexlab{a}})Wang, Jia, Cao, and Gong]{wang2021certified}
Binghui Wang, Jinyuan Jia, Xiaoyu Cao, and Neil~Zhenqiang Gong.
\newblock Certified robustness of graph neural networks against adversarial structural perturbation.
\newblock In \emph{KDD}, 2021{\natexlab{a}}.

\bibitem[Wang et~al.(2022{\natexlab{a}})Wang, Li, and Zhou]{wang2022bandits}
Binghui Wang, Youqi Li, and Pan Zhou.
\newblock Bandits for structure perturbation-based black-box attacks to graph neural networks with theoretical guarantees.
\newblock In \emph{CVPR}, 2022{\natexlab{a}}.

\bibitem[Wang et~al.(2023{\natexlab{a}})Wang, Pang, and Dong]{wang2023turning}
Binghui Wang, Meng Pang, and Yun Dong.
\newblock Turning strengths into weaknesses: A certified robustness inspired attack framework against graph neural networks.
\newblock In \emph{CVPR}, 2023{\natexlab{a}}.

\bibitem[Wang et~al.(2024)Wang, Lin, Zhou, Zhou, Li, Pang, Li, and Chen]{wang2024efficient}
Binghui Wang, Minhua Lin, Tianxiang Zhou, Pan Zhou, Ang Li, Meng Pang, Hai Li, and Yiran Chen.
\newblock Efficient, direct, and restricted black-box graph evasion attacks to any-layer graph neural networks via influence function.
\newblock In \emph{WSDM}, 2024.

\bibitem[Wang et~al.(2023{\natexlab{b}})Wang, Huang, Chandak, Zitnik, and Gehlenborg]{Drug_repurposing2023}
Qianwen Wang, Kexin Huang, Payal Chandak, Marinka Zitnik, and Nils Gehlenborg.
\newblock Extending the nested model for user-centric xai: A design study on gnn-based drug repurposing.
\newblock \emph{IEEE TVCG}, 2023{\natexlab{b}}.

\bibitem[Wang et~al.(2023{\natexlab{c}})Wang, Yin, Li, Xie, and Wang]{wang2023vinfor}
Senzhang Wang, Jun Yin, Chaozhuo Li, Xing Xie, and Jianxin Wang.
\newblock V-infor: A robust graph neural networks explainer for structurally corrupted graphs.
\newblock In \emph{NeurIPS}, 2023{\natexlab{c}}.

\bibitem[Wang et~al.(2022{\natexlab{b}})Wang, Levine, and Feizi]{wang2022improved}
Wenxiao Wang, Alexander~J Levine, and Soheil Feizi.
\newblock Improved certified defenses against data poisoning with (deterministic) finite aggregation.
\newblock In \emph{ICML}, pp.\  22769--22783. PMLR, 2022{\natexlab{b}}.

\bibitem[Wang et~al.(2021{\natexlab{b}})Wang, Wu, Zhang, He, and Chua]{wang2021towards}
Xiang Wang, Yingxin Wu, An~Zhang, Xiangnan He, and Tat-Seng Chua.
\newblock Towards multi-grained explainability for graph neural networks.
\newblock In \emph{NeurIPS}, volume~34, pp.\  18446--18458, 2021{\natexlab{b}}.

\bibitem[Wang et~al.(2023{\natexlab{d}})Wang, Wu, Zhang, Feng, He, and Chua]{Wang_2023/RCExplainer}
Xiang Wang, Yingxin Wu, An~Zhang, Fuli Feng, Xiangnan He, and Tat-Seng Chua.
\newblock Reinforced causal explainer for graph neural networks.
\newblock \emph{{IEEE} TPAMI}, pp.\  2297--2309, 2023{\natexlab{d}}.

\bibitem[Wang \& Shen(2023)Wang and Shen]{wang2023/gnninterpreter}
Xiaoqi Wang and Han~Wei Shen.
\newblock {GNNI}nterpreter: A probabilistic generative model-level explanation for graph neural networks.
\newblock In \emph{ICLR}, 2023.

\bibitem[Wu et~al.(2019)Wu, Zhang, Souza~Jr, Fifty, Yu, and Weinberger]{wu2019simplifying}
Felix Wu, Tianyi Zhang, Amauri Holanda~de Souza~Jr, Christopher Fifty, Tao Yu, and Kilian~Q Weinberger.
\newblock Simplifying graph convolutional networks.
\newblock In \emph{ICML}, 2019.

\bibitem[Wu et~al.(2023)Wu, Wang, Du, Jiang, Kang, Li, Pan, Deng, Cao, Hsieh, et~al.]{wu2023chemistry}
Zhenxing Wu, Jike Wang, Hongyan Du, Dejun Jiang, Yu~Kang, Dan Li, Peichen Pan, Yafeng Deng, Dongsheng Cao, Chang-Yu Hsieh, et~al.
\newblock Chemistry-intuitive explanation of graph neural networks for molecular property prediction with substructure masking.
\newblock \emph{Nature Communications}, 14\penalty0 (1):\penalty0 2585, 2023.

\bibitem[Xia et~al.(2021)Xia, Lee, Bengio, and Bareinboim]{xia2021causal}
Kevin Xia, Kai-Zhan Lee, Yoshua Bengio, and Elias Bareinboim.
\newblock The causal-neural connection: Expressiveness, learnability, and inference.
\newblock In \emph{NeurIPS}, 2021.

\bibitem[Xia et~al.(2024)Xia, Yang, Wang, and Jia]{xia2024gnncert}
Zaishuo Xia, Han Yang, Binghui Wang, and Jinyuan Jia.
\newblock Gnncert: Deterministic certification of graph neural networks against adversarial perturbations.
\newblock In \emph{ICLR}, 2024.

\bibitem[Xiang et~al.(2021)Xiang, Bhagoji, Sehwag, and Mittal]{xiang2021patchguard}
Chong Xiang, Arjun~Nitin Bhagoji, Vikash Sehwag, and Prateek Mittal.
\newblock $\{$PatchGuard$\}$: A provably robust defense against adversarial patches via small receptive fields and masking.
\newblock In \emph{USENIX Security}, 2021.

\bibitem[Xu et~al.(2019)Xu, Chen, Liu, Chen, Weng, Hong, and Lin]{xu2019topology}
Kaidi Xu, Hongge Chen, Sijia Liu, Pin-Yu Chen, Tsui-Wei Weng, Mingyi Hong, and Xue Lin.
\newblock Topology attack and defense for graph neural networks: An optimization perspective.
\newblock In \emph{IJCAI}, 2019.

\bibitem[{Xu} et~al.(2019){Xu}, {Hu}, {Leskovec}, and {Jegelka}]{xu2018how}
Keyulu {Xu}, Weihua {Hu}, Jure {Leskovec}, and Stefanie {Jegelka}.
\newblock How powerful are graph neural networks?
\newblock In \emph{ICLR}, 2019.

\bibitem[Xu et~al.(2019)Xu, Hu, Leskovec, and Jegelka]{xu2018powerful}
Keyulu Xu, Weihua Hu, Jure Leskovec, and Stefanie Jegelka.
\newblock How powerful are graph neural networks?
\newblock In \emph{ICLR}, 2019.
\newblock URL \url{https://openreview.net/forum?id=ryGs6iA5Km}.

\bibitem[Yang et~al.(2024)Yang, Li, Jia, Hong, and Wang]{yang2024distributed}
Yuxin Yang, Qiang Li, Jinyuan Jia, Yuan Hong, and Binghui Wang.
\newblock Distributed backdoor attacks on federated graph learning and certified defenses.
\newblock In \emph{CCS}, 2024.

\bibitem[Yang et~al.(2022)Yang, Zhong, Zhao, and Chen]{yang2022mgraphdta}
Ziduo Yang, Weihe Zhong, Lu~Zhao, and Calvin Yu-Chian Chen.
\newblock Mgraphdta: deep multiscale graph neural network for explainable drug--target binding affinity prediction.
\newblock \emph{Chemical science}, 13\penalty0 (3):\penalty0 816--833, 2022.

\bibitem[Ying et~al.(2019)Ying, Bourgeois, You, Zitnik, and Leskovec]{GNNEx19}
Zhitao Ying, Dylan Bourgeois, Jiaxuan You, Marinka Zitnik, and Jure Leskovec.
\newblock {GNNE}xplainer: Generating explanations for graph neural networks.
\newblock In \emph{NeurIPS}. 2019.

\bibitem[Yuan et~al.(2021)Yuan, Yu, Wang, Li, and Ji]{DBLP:journals/corr/abs-2102-05152/subgraphX}
Hao Yuan, Haiyang Yu, Jie Wang, Kang Li, and Shuiwang Ji.
\newblock On explainability of graph neural networks via subgraph explorations.
\newblock In \emph{ICML}, 2021.

\bibitem[Zeng et~al.(2020)Zeng, Zhou, Srivastava, Kannan, and Prasanna]{Zeng2020GraphSAINT}
Hanqing Zeng, Hongkuan Zhou, Ajitesh Srivastava, Rajgopal Kannan, and Viktor Prasanna.
\newblock Graphsaint: Graph sampling based inductive learning method.
\newblock In \emph{ICLR}, 2020.

\bibitem[Zhang et~al.(2023)Zhang, Jia, Liu, and Gong]{zhang2023pointcert}
Jinghuai Zhang, Jinyuan Jia, Hongbin Liu, and Neil~Zhenqiang Gong.
\newblock Pointcert: Point cloud classification with deterministic certified robustness guarantees.
\newblock In \emph{CVPR}, 2023.

\bibitem[Zhang et~al.(2022)Zhang, Liu, Shah, and Sun]{zhang2022gstarx}
Shichang Zhang, Yozen Liu, Neil Shah, and Yizhou Sun.
\newblock Gstarx: Explaining graph neural networks with structure-aware cooperative games.
\newblock In \emph{NeurIPS}, volume~35, pp.\  19810--19823, 2022.

\bibitem[Zhang et~al.(2021{\natexlab{a}})Zhang, Jia, Wang, and Gong]{zhang2020backdoor}
Zaixi Zhang, Jinyuan Jia, Binghui Wang, and Neil~Zhenqiang Gong.
\newblock Backdoor attacks to graph neural networks.
\newblock 2021{\natexlab{a}}.

\bibitem[Zhang et~al.(2021{\natexlab{b}})Zhang, Jia, Wang, and Gong]{zhang2021backdoor}
Zaixi Zhang, Jinyuan Jia, Binghui Wang, and Neil~Zhenqiang Gong.
\newblock Backdoor attacks to graph neural networks.
\newblock In \emph{SACMAT}, 2021{\natexlab{b}}.

\bibitem[Zhu et~al.(2023)Zhu, Wen, Song, Ma, and Wang]{zhu2023hierarchical}
Wenhao Zhu, Tianyu Wen, Guojie Song, Xiaojun Ma, and Liang Wang.
\newblock Hierarchical transformer for scalable graph learning.
\newblock In \emph{IJCAI}, 2023.

\bibitem[Z{\"u}gner \& G{\"u}nnemann(2019)Z{\"u}gner and G{\"u}nnemann]{zugner2019adversarial}
Daniel Z{\"u}gner and Stephan G{\"u}nnemann.
\newblock Adversarial attacks on graph neural networks via meta learning.
\newblock In \emph{ICLR}, 2019.

\bibitem[Z{\"u}gner et~al.(2018)Z{\"u}gner, Akbarnejad, and G{\"u}nnemann]{zugner2018adversarial}
Daniel Z{\"u}gner, Amir Akbarnejad, and Stephan G{\"u}nnemann.
\newblock Adversarial attacks on neural networks for graph data.
\newblock In \emph{SIGKDD}, pp.\  2847--2856, 2018.

\end{thebibliography}
